\documentclass[onefignum,onetabnum]{siamart190516}

\usepackage{amssymb,times,enumerate,amsmath,subfigure,graphicx,color,ifthen,amsfonts}
\usepackage{algorithm,algpseudocode,stmaryrd,newfile,mathtools,bm}
\algtext*{EndFor}% Remove "end if" text
\usepackage{mathtools}
\usepackage{paralist}
\usepackage{cite}
\usepackage{hyperref}

\DeclareMathOperator*{\rank}{rank}

\numberwithin{table}{section}%{section}
\numberwithin{equation}{section}%{section}
\numberwithin{algorithm}{section}%{section}

\newcommand{\purple}[1]{\textcolor{violet}{#1}}
\newoutputstream{fixmes}
\openoutputfile{fixme.txt}{fixmes}
\newcounter{fixcounter}
\numberwithin{fixcounter}{section}
\newcommand{\fixme}[2]{\stepcounter{fixcounter} \purple{FIXME \thefixcounter \  (#1): #2}   \addtostream{fixmes}{\noexpand FIXME \thefixcounter  (#1): #2}}
\newcommand{\fixed}[2]{\stepcounter{fixcounter} \addtostream{fixmes}{\noexpand FIXME \thefixcounter: #1: #2}}
\definecolor{darkgreen}{rgb}{0.2,0.6,0.2}
\definecolor{darkred}{rgb}{0.6,0.2,0.2}

\newcommand{\mathsmall}{}
\newcommand{\mathfootnotesize}{}
\def \ISTHESIS {FALSE}
\def \includecommcost {FALSE}

\newcommand*{\cachesize}{\ensuremath{H}}
\newcommand*{\cs}{\ensuremath{H}}
\newcommand*{\bwcost}{\ensuremath{Q}}
\newcommand{\commcost}[1]{\ifthenelse{\equal{\includecommcost}{TRUE}}{#1}{}}
\newcommand{\notcommcost}[1]{\ifthenelse{\equal{\includecommcost}{FALSE}}{#1}{}}
\newcommand{\onlyinthesis}[1]{\ifthenelse{\equal{\ISTHESIS}{TRUE}}{#1}{}}
\newcommand{\notinthesis}[1]{\ifthenelse{\equal{\ISTHESIS}{TRUE}}{}{#1}}
\newtheorem{algorithmdef}{Algorithm}[section]
\newcommand{\nsctr}[5]{{{#1}} \otimes_#5 {{#2}}}
\newcommand{\syctr}[5]{{{#1}} \circ_#5 {{#2}}}

\newcommand{\nsalg}[5]{\Upsilon^{(#3,#4,#5)}({{#1}},{{#2}})}
\newcommand{\bilalg}[1]{\Lambda^{(\mathcal{#1})}}

\newcommand{\systalg}[5]{\Psi_\circ^{(#3,#4,#5)}({{#1}},{{#2}})}

\newcommand{\syfsalg}[5]{\Phi_\circ^{(#3,#4,#5)}({{#1}},{{#2}})}

\newcommand{\nnsalg}[3]{\Upsilon^{(#1,#2,#3)}}
\newcommand{\nsystalg}[3]{\Psi_\circ^{(#1,#2,#3)}}

\newcommand{\nsyfsalg}[3]{\Phi_\circ^{(#1,#2,#3)}}

\newcommand{\nbilalg}[5]{\hat{\Lambda}_{(\mathcal{#4},\mathcal{#5})}^{(#1,#2,#3)}}

\newcommand{\getmat}[1]{\langle{#1}\rangle}

\newcommand{\BB}[1]{\B{\bar{#1}}}

\newcommand{\lt}{\left}
\newcommand{\rt}{\right}

\newcommand{\tpl}[2]{{\B{#1}}}
\newcommand{\inti}[2]{[{#1},{#2}]}
\newcommand{\bunderline}[1]{\mkern1mu\underline{\mkern-1mu#1\mkern-1mu}\mkern1mu }
\newcommand{\enti}[3]{\bunderline{\angle}[{#1},{#2}]^{#3}}
\newcommand{\lnti}[3]{\angle[{#1},{#2}]^{#3}}
\newcommand{\chchoose}[2]{\lt(\!\!\binom{#1}{#2}\!\!\rt)}
\newcommand{\tchchoose}[2]{\big(\!\!\binom{#1}{#2}\!\!\big)}
\newcommand{\seqarch}[1]{\mathfrak{N}(#1)}
\newcommand{\pararch}[2]{\mathfrak{M}(#1)}

\newcommand{\Ce}[1]{\lowercase{#1}}
\newcommand{\C}[1]{\bm{\mathcal{#1}}}

\newcommand{\bFA}{\B{{F}}^{(A)}}
\newcommand{\bFB}{\B{{F}}^{(B)}}
\newcommand{\bFC}{\B{{F}}^{(C)}}

\newcommand{\bhRA}{\B{{R}}^{(A)}}

\newcommand{\bhRB}{\B{{R}}^{(B)}}

\newcommand{\bhRC}{\C{{R}}^{(C)}}
\newcommand{\defeq}{\coloneqq}
\newcommand{\B}{\bm}

\newcommand{\bprt}[2]{\bar{\chi}^{#1}_{#2}}
\newcommand{\prt}[2]{\chi^{#1}_{#2}}
\newcommand{\prj}[1]{\chi^{#1}}

\newcommand{\perms}[1]{\mathfrak{G}_{#1}}
\newcommand{\rev}[1]{#1}

\title{Communication Lower Bounds of Bilinear Algorithms for Symmetric Tensor Contractions
}

\author{Edgar Solomonik\thanks{Department of Computer Science, University of Illinois at Urbana-Champaign, USA} 
  (\email{solomon2@illinois.edu})
\and James Demmel\thanks{Department of EECS (Computer Science Division) and Department of Mathematics, University of California, Berkeley, USA}
  (\email{demmel@cs.berkeley.edu})
\and Torsten Hoefler\thanks{Department of Computer Science, ETH Zurich, Switzerland} (\email{htor@inf.ethz.ch})}

\begin{document}
\maketitle

\begin{abstract}
We introduce a new theoretical framework for deriving lower bounds on data movement in bilinear algorithms. Bilinear algorithms are a general representation of fast algorithms for bilinear functions, which include computation of matrix multiplication, convolution, and symmetric tensor contractions. A bilinear algorithm is described by three matrices. Our communication lower bounds are based on quantifying the minimal matrix ranks of matching subsets of columns of these matrices. This infrastructure yields new communication lower bounds for symmetric tensor contraction algorithms, which provide qualitative new insights. 
%Tensor contraction operations constitute the main computational bottleneck within numerical methods for modelling interactions between electrons in chemical systems.
\rev{Tensor symmetry (invariance under permutation of modes) is common to many applications of tensor computations (e.g., tensor representation of hypergraphs, analysis of high order moments in data, as well as tensors modelling interactions of electrons in computational chemistry).}
%Symmetric tensor contractions operations constitute the main computational bottleneck within numerical methods for modelling interactions between electrons in chemical systems.
%In these calculations, tensors are often symmetric (invariant to certain interchanges of modes), enabling 
\rev{Tensor symmetry enables reduction in representation size as well as arithmetic cost of contractions}
by factors that scale with the number of equivalent permutations. However, we derive lower bounds showing that these arithmetic cost and memory reductions can necessitate increases in data movement by factors that scale with the size of the tensors.

\end{abstract}

\section{Introduction}
\label{sec:intro}

Tensor contractions are tensor products that are summed (contracted) over a subset of the modes (indices).
They generalize the product of a matrix and a vector and of two matrices.
Symmetric tensors, like symmetric matrices, are invariant under permutation (transposition) of modes.
\rev{Tensor symmetry provides the opportunity to reduce memory footprint by storing only the unique elements in the tensors (e.g., storing only the upper triangular part of a symmetric matrix), a technique we refer to as {\it packed storage}.
Symmetry also enables reduced computational cost for tensor contractions, but this requires algebraic reorganization of the tensor contraction, which we refer to as the {\it symmetry preserving algorithm}~\cite{solomonikfast}.
A consequence of this reorganization is that the contraction can no longer be reduced to a product of matrices or vectors, a standard technique for symmetric and nonsymmetric tensor contractions~\cite{di2014towards,JCC:JCC23377,SMHD_IPDPS_2013,Rajbhandari:2014:CFC:2683593.2683635}, which we refer to as the {\it direct evaluation algorithm}.
Both of these approaches can be expressed using the formalism of {\it bilinear algorithms}~\cite{pan1984can}, a representation of fast algorithms for bilinear problems like matrix multiplication, convolution, and (symmetric) tensor contraction.
By applying a new framework for deriving communication lower bounds for bilinear algorithms to symmetric tensor contraction algorithms, we show that, relative to known upper bounds for the direct evaluation algorithm, the use of either packed storage or the symmetry preserving algorithm can induce asymptotic overhead in communication cost in certain settings.

Symmetric tensors arise in a variety of applications, including electronic structure methods in computational chemistry~\cite{Bartlett:1981:ARPC:CC}, analysis of higher-order moments in data~\cite{sherman2020estimating}, and Laplacian tensors of hypergraphs~\cite{chen2017fiedler}.
The use of packed storage and symmetry preserving tensor contraction algorithms is particularly beneficial when the tensor order (number of modes in the tensor) is high (the number of permutations grows factorially with the number of symmetric modes in the tensor).
Such cases are common in high-accuracy methods for electronic structure, such as the coupled cluster family of methods~\cite{Bartlett:1981:ARPC:CC}.
The symmetry preserving algorithm has been shown to reduce cost of variants of these methods by factors of 2-9X (the higher speed-ups achieved for higher-order methods, as these leverage higher-order tensors)~\cite{SD_ETHZ_2015}.
}

%In electronic structure methods, tensor representations allow expression of different types of electron--orbital interactions as different contractions.
%Tensors representing multi-electron and multi-orbital interactions encode interchangeability of these entities as permutational symmetries in the data representation~\cite{Bartlett:1981:ARPC:CC}.
%While any contraction may be reduced to a product of two matrices (or vectors), such reductions are suboptimal in computation cost for tensors with symmetry.
%Algebraic reorganization of the contractions~\cite{solomonikfast} can yield reductions in computation cost ranging from 2X for contractions in common methods such as CCSD to 4X and 9X for contractions in higher-order methods (CCSDT and CCSDTQ respectively)~\cite{SD_ETHZ_2015}.
%We aim to compare the communication cost of reducing the contractions of symmetric tensors to matrix multiplication (the {\it direct evaluation algorithm}) to executing the algebraically reorganized form (the {\it symmetry preserving algorithm}).

%We formalize the problem of a 
Our analysis targets the problem of
{\it symmetric contraction}, the contraction of two symmetric tensors and subsequent symmetrization of the result.
Symmetrization involves taking the sum of all permutations of the initial contraction product and results in a symmetric tensor (e.g., computing $\B A \B B + \B B \B A$, where $\B A$ and $\B B$ are symmetric matrices).
Electronic structure methods employ symmetrization, but additionally also involve antisymmetric tensors, antisymmetrization, as well as partial symmetries.
We consider exclusively symmetric tensors~\cite{comon2008symmetric} (otherwise referred to as fully symmetric or supersymmetric~\cite{doi:10.1137/07070111X}).
The relevant algorithms for antisymmetric tensors are similar in structure~\cite{SD_ETHZ_2015}, but extension of our results \rev{for symmetry preserving algorithms} to partially symmetric tensors is nontrivial.
\rev{
For partially symmetric tensors, symmetry preserving algorithms can be applied in alternative ways for a single contraction, and should be nested with standard contraction algorithms, resulting in bilinear algorithms with a more complicated structure.}
%Extension of our analysis to capture usage of the symmetric preserving algorithm with partially-symmetric tensors~\cite{SD_ETHZ_2015} is non-trivial and left as future work.

%In this paper, we study the performance characteristics of the classical `direct evaluation algorithm'
%and new `symmetry preserving algorithm'~\cite{SD_ETHZ_2015} for `fully symmetric contractions' (symmetrized contractions of symmetric tensors).
%We demonstrate that despite its lower computation cost, for some contractions, the symmetry preserving algorithm requires asymptotically 
%more vertical (between memory and cache) and horizontal (interprocessor) communication than the direct evaluation approach.
%Further, we demonstrate that for some fully symmetric contractions, the direct evaluation algorithm requires asymptotically more horizontal communication
%than multiplication of matrices of corresponding size.

%Our main result is a lower bound on the communication cost of the `symmetry preserving' approach, which demonstrates that for certain contractions
%the new approach requires less work, yet asymptotically more communication with respect to just employing a matrix multiplication.

%We characterize the communication requirements of both of these algorithms by deriving lower bounds on the communication cost on any schedule which executes the algorithms.
\rev{Bilinear algorithms are defined} by a set of products of linear combinations of \rev{two sets of} inputs and partial summations of these bilinear forms~\cite{pan1984can}.
% and characterize the complete space of possible algorithms algebraically.
\rev{A bilinear algorithm is described by three sparse matrices, each defining what linear combinations need to be computed for either one of the input sets or partial summations for the output.
The number of products computed by the bilinear algorithm is referred to as its {\it (bilinear) rank} and is a measure of the complexity of the algorithm.
}
%Each of these algorithms can be encoded by a set of three sparse matrices as a {\it bilinear algorithm}.
We derive the communication requirements of any bilinear algorithm 
\rev{based on an expansion bound, i.e., an algorithm-specific function relating the ranks of submatrices of the three sparse matrices.}
Given this infrastructure, it suffices to derive expansion bounds for each bilinear algorithm to obtain communication lower bound results.
We leverage a generalization of the Loomis-Whitney inequality~\cite{loomis1949n,Tiskin98thedesign} to prove expansion bounds based on the sparsity structure of the matrices encoding each algorithm.
% is then applied to each of the particular bilinear algorithms we study.
%The direct evaluation algorithm is defined by the same set of products and summations as a related matrix multiplication as well as some additional summations for symmetrization.
%The symmetry preserving algorithm is defined as a set of products of sums of tensor elements and their subsequent accumulation.
%These definitions allow us to analyze well-defined algorithmic dependency structures.
%However, the algorithm descriptions do not specify the summation order and consequently avoid restricting the order of operations.

%
\rev{Communication cost lower bounds make it possible to ascertain optimality of schedules for a particular algorithm or family of algorithms.}
We associate an {\it execution DAG} (directed acyclic graph) with a particular specification of a bilinear algorithm (multiple execution DAGs are possible for a single bilinear algorithm).
We then consider potential sequential schedules and parallel schedules of execution DAGs.
We measure the {\it vertical communication cost} of a sequential schedule as the amount of data moved between memory and a single level of cache.
We measure the {\it horizontal communication cost} of a parallel schedule as the amount of data communicated between any one processor and the others.
Our horizontal communication lower bounds can be translated to the LogGP~\cite{Alexandrov:1995} model and the bulk synchronous parallel (BSP) model~\cite{valiant1990bridging}.

The topic of communication lower bounds has been studied extensively for numerical linear algebra operations.
For classical matrix multiplication, results have yielded lower bounds on vertical communication~\cite{Jia-Wei:1981:ICR:800076.802486} and horizontal communication~\cite{aggarwal1989communication,Tiskin98thedesign,irony_mm_lb04}.
These have been extended to Strassen's algorithm for matrix multiplication~\cite{Ballard:SPAA2011}, matrix factorizations~\cite{greygeneral2010}, and the matricized-tensor times Khatri-Rao product (MTTKRP)~\cite{ballard2018communication}.
These communication lower bound results work by analyzing a particular execution DAG or family of execution DAGs.

Using bilinear algorithms, we provide a framework for communication lower bound derivation for an algebraic specification of the algorithm, which automatically applies to a family of execution DAGs.
%Our bilinear algorithm representation and rank-based communication lower-bounds are more concise, powerful, and extensible than previous models.
By working with the rank structure of the matrices encoding a bilinear algorithm, we obtain lower bounds that apply to any \rev{possible execution DAG} for the bilinear algorithm, including ones that use linear transformations to compress data so as to reduce communication.
The bilinear algorithm framework enables us to reproduce existing results as well as to obtain new lower bounds for symmetric tensor contractions, which are not amenable to existing techniques.
%To the best of our knowledge, our work is the first to consider the effect of tensor symmetry on communication lower bounds for tensor contractions.
We give lower bounds on vertical and horizontal communication costs of any sequential or parallel schedule for the direct evaluation and symmetry preserving algorithms for symmetric tensor contractions.

\rev{We characterize symmetric tensor contractions using the integer parameters $n,s,t,v$, where $n$ is the dimension (size of index range) of each mode of the tensor, $s$ is the number of uncontracted modes in the first input tensor, $t$ is the number of uncontracted modes in the second input tensor, and $v$ is the number of contracted modes.
%The dimension (range of indices) of all modes of the tensors is given by $n$.
For example, the product of a symmetric matrix and a vector is given by $s{=}1,t{=}0,v{=}1$, a symmetrized product of vectors $\B a \B b^T + \B b \B a^T$ is given by $s{=}1,t{=}1,v{=}0$, while $s{=}t{=}v{=}1$ corresponds to the symmetrized product of two symmetric matrices $\B A \B B + \B B \B A$.
When one of $s$, $t$, or $v$ is zero, we say the contraction is {\it matrix-vector-like} because of its form in the nonsymmetric case, e.g., if $t=0$ it reduces to the product of an $n^s\times n^v$ matrix with a vector of dimension $n^v$ (when $v=0$ it reduces to an outer product).
Otherwise (if $s,t,v\geq 1$), we say the contraction is {\it matrix-matrix-like} (reduces to the product of an $n^s\times n^v$ matrix with a $n^v\times n^t$ matrix).

We consider existing characterizations of asymptotic communication cost (which provide matching lower and upper bounds for vertical and horizontal communication) for matrix-vector and matrix-matrix products (assuming the classical $O(n^3)$ algorithm for $n\times n$ matrices as opposed to Strassen's algorithm or another fast matrix multiplication algorithm).
With these existing results in mind, the lower bounds introduced in this paper yield the following theoretical contributions.
\begin{itemize}
\item
Packed storage necessitates asymptotically more horizontal communication (relative to storing the tensors as if they were nonsymmetric) for the direct evaluation algorithm for some matrix-vector-like contractions. 
Specifically, this overhead arises for contractions in which one of $s$, $t$, or $v$ is zero and the remaining two are unequal (e.g., for $s=2,t=0,v=1$, which in the nonsymmetric case corresponds to a product of an $n^{2}\times n$ matrix and a vector).
\item
The symmetry preserving contraction algorithm requires asymptotically more horizontal and vertical communication than the direct evaluation algorithm for some matrix-matrix-like contractions, in particular any $s,t,v\geq 1$ except when $s=t=v$.
\item
We derive lower bounds for nonsymmetric and symmetric contractions by application of a new general framework, which captures a family of execution DAGs for a bilinear algorithm automatically, and may be useful for deriving communication lower bounds for other bilinear algorithms.
\end{itemize}
}

We summarize the communication lower bound results obtained for different tensor contraction algorithms in Table~\ref{tab:lbsum}.
The first three columns ($s$, $t$, $v$) give the order of the tensor contraction operands: $s+v$ and $v+t$, as well as of the result, $s+t$.
The dimension (range of indices) of all modes of the tensors is given by $n$.
For example, row 3 ($s{=}1,t{=}1,v{=}0$) corresponds to the symmetrized outer product $\B a \B b^T + \B b \B a^T$ (which has the same costs as the product of a symmetric matrix and a vector, for which $s{=}1,t{=}0,v{=}1$), while row 8 ($s{=}t{=}v{=}1$) corresponds to the symmetrized product of two symmetric matrices ($\B A \B B + \B B \B A$).
The following three columns ($F_\Upsilon$, $F_\Psi$, $F_\Phi$) provide the number of multiplications (to leading order in $n$) needed for nonsymmetric contractions, the direct evaluation algorithm for symmetric contractions, and the symmetry preserving algorithm, respectively.
The arithmetic costs $F_\Phi$ is shown in green, because it can be much smaller than the arithmetic costs of the other algorithms (by a factor of up to $\omega! = 720$).

The last five columns in Table~\ref{tab:lbsum} list the asymptotic communication cost lower bounds (\rev{denoted by $f(n)$ instead of $\Omega(f(n))$ for brevity}), where $\bwcost$ is vertical communication ($\cs$ is the cache size) and $W$ is horizontal communication (with $p$ processors).
The vertical communication cost lower bounds for the nonsymmetric case and the direct evaluation algorithm are listed within one column, because \rev{these} are always asymptotically equivalent.
The table presents the lower bounds in simplified form with the use of the assumptions: $n\geq p\gg 1$ and $\cs \leq  n^2$.
These assumptions are sensible when $\omega=s+t+v$ is small.
The communication costs are shown in red when they are higher than for the nonsymmetric algorithm.
%The columns labeled `rank', correspond to the bilinear rank of each algorithm, while the communication cost columns list the asymptotic communication complexity.
%The table lists examples for a variety of $s,t,v$ values.

\begin{table*}[ht]
\centering
\caption{The table presents bilinear rank ($F$) and communication cost lower bounds ($\bwcost$--vertical and $W$--horizontal) for nonsymmetric and symmetric tensor contraction algorithms ($\Upsilon$--nonsymmetric, $\Psi$--direct evaluation, $\Phi$--symmetry preserving). 
Row 2 gives the general costs for matrix-vector-like contractions, and rows 3-6 list particular instances. 
Row 7 gives the general cost for matrix-matrix-like contractions, and rows 8-11 list particular instances.
The results are symmetric in permutation of $(s,t,v)$, so we focus on $s\geq t\geq v$.
Green coloring shows where improvements obtained via symmetry, while red coloring identifies overheads of exploiting it.
}
\renewcommand{\arraystretch}{1.5}
{
\footnotesize
\begin{tabular}{c|c|c||c|c|c||c|c||c|c|c}
%\multicolumn{3}{c}{$\nbilalg stvDF$} &
%\multicolumn{3}{c}{$\nnsalg stv$} &
%\multicolumn{3}{c}{$\nsystalg stv$}  & 
%\multicolumn{3}{c}{$\nsyfsalg stv$} \\
%\hline
$s$ & $t$ & $v$ & 
$F_\Upsilon$ & $F_\Psi$  & $F_\Phi$ & $\bwcost_{\Upsilon,\Psi}$ & $\bwcost_\Phi$ &  $W_\Upsilon$ & $W_\Psi$ & $W_\Phi$ 
\\
\hline
\hline
$\geq t$ & $\geq 1$ & $0$ &
$n^{\omega}$ & $\frac{n^{\omega}}{s!t!}$ & $\textcolor{darkgreen}{\frac{n^{\omega}}{\omega!}}$ &
$n^{\omega}$ & $n^{\omega}$ &
$\min\big(n^t,\frac{n^{\omega/2}}{p^{1/2}}\big)$ & $\textcolor{darkred}{\frac{n^{s}}{p^{s/\omega}}}$ & $\textcolor{darkred}{\frac{n^{s}}{p^{s/\omega}}}$
\\[.5ex] 
\hline
\hline
1 & 1 & 0 &
$n^2$ & $n^2$ & $\textcolor{darkgreen}{\frac{n^2}2}$ &
$n^2$ & $n^2$ & 
$\frac{n}{p^{1/2}}$ & $\frac{n}{p^{1/2}}$ & $\frac{n}{p^{1/2}}$  
\\[.5ex]
\hline
2 & 1 & 0 &
$n^3$ &  $\frac{n^3}{2}$ &  $\textcolor{darkgreen}{\frac{n^3}{6}}$ &
$n^3$ & $n^3$ & 
$n$ %$\frac{n^{3/2}}{p^{1/2}}$ 
& $\textcolor{darkred}{\frac{n^{2}}{p^{2/3}}}$  & $\textcolor{darkred}{\frac{n^2}{p^{2/3}}}$  
\\[.5ex] 
\hline
3 & 1 & 0 &
$n^4$ & $\frac{n^4}{6}$ & $\textcolor{darkgreen}{\frac{n^4}{24}}$ &
$n^4$ & $n^4$ &
$n$ %$\frac{n^{2}}{p^{1/2}}$ 
& $\textcolor{darkred}{\frac{n^{3}}{p^{3/4}}}$ & $\textcolor{darkred}{\frac{n^3}{p^{3/4}}}$ 
\\[.5ex]
\hline
2 & 2 & 0 &
$n^4$ & $\frac{n^4}{4}$ & $\textcolor{darkgreen}{\frac{n^4}{24}}$ &
$n^4$ & $n^4$ & 
$\frac{n^2}{p^{1/2}}$ & $\frac{n^2}{p^{1/2}}$ & $\frac{n^2}{p^{1/2}}$  
\\[.5ex] 
\hline
\hline
${\geq}t$ & ${\geq}v$ & ${\geq}1$ &
$n^{\omega}$ & $\frac{n^{\omega}}{s!t!v!}$ & $\textcolor{darkgreen}{\frac{n^{\omega}}{\omega!}}$ &
$\frac{n^{\omega}}{\cs^{1/2}}$ & $\textcolor{darkred}{\frac{n^\omega\cs}{\cs^{\frac{\omega}{s+t}}}}$ &
\multicolumn{2}{c|}{$
\min\Big(
n^{t{+}v}{,}
\frac{n^{\frac{s{+}t}{2}{+}v}}{p^{1/2}}{,}
\frac{n^{\frac{2\omega}3}}{p^{2/3}}
\Big)$ } & %$\frac{n^{2\omega/3}}{p^{2/3}}$ & 
$\textcolor{darkred}{\frac{n^{s{+}t}}{p^{\frac{s{+}t}\omega}}}$
\\[1ex] 
\hline
\hline
1 & 1 & 1 &
$n^3$ & $n^3$ & $\textcolor{darkgreen}{\frac{n^3}{6}}$ & 
$\frac{n^3}{\cs^{1/2}}$ & $\frac{n^3}{\cs^{1/2}}$ & 
$\frac{n^2}{p^{2/3}}$ & $\frac{n^2}{p^{2/3}}$ & $\frac{n^2}{p^{2/3}}$  
\\[.5ex] 
\hline
2 & 1 & 1 &
$n^4$ & $\frac{n^4}{2}$ & $\textcolor{darkgreen}{\frac{n^4}{24}}$ &
$\frac{n^4}{\cs^{1/2}}$ & $\textcolor{darkred}{\frac{n^4}{\cs^{1/3}}}$ &
$n^2$ & $n^2$ & %$\frac{n^{8/3}}{p^{2/3}}$ & $\frac{n^{8/3}}{p^{2/3}}$ 
$\textcolor{darkred}{\frac{n^3}{p^{3/4}}}$  
\\[.5ex] 
\hline
2 & 2 & 1 & 
$n^5$ & $\frac{n^5}{4}$ & $\textcolor{darkgreen}{\frac{n^5}{120}}$ &
$\frac{n^5}{\cs^{1/2}}$ & $\textcolor{darkred}{\frac{n^5}{\cs^{1/4}}}$ &
$\frac{n^3}{p^{1/2}}$ & $\frac{n^3}{p^{1/2}}$ & %$\frac{n^{10/3}}{p^{2/3}}$ & $\frac{n^{10/3}}{p^{2/3}}$ & 
$\textcolor{darkred}{\frac{n^4}{p^{4/5}}}$  
\\[.5ex] 
\hline
2 & 2 & 2 &
$n^6$ & $\frac{n^6}{8}$ & $\textcolor{darkgreen}{\frac{n^6}{720}}$ &
$\frac{n^6}{\cs^{1/2}}$ &  $\frac{n^6}{\cs^{1/2}}$ &
$\frac{n^{4}}{p^{2/3}}$ & $\frac{n^{4}}{p^{2/3}}$ & $\frac{n^4}{p^{2/3}}$  
\end{tabular}
}
\label{tab:lbsum}
\end{table*}

\section{Tensor Notation} % and Definitions}
\label{sec:defs}

%We employ the notational constructs used in~\cite{SD_ETHZ_2015}. 
%Below we review the essentials of the notation defined in that paper in less detail and additional restrictions appropriate for this work.
We use the notation from the introductory work on the symmetry preserving algorithm~\cite{solomonikfast} with some modifications.
We additionally restrict all elements of tensors to be in the same algebraic ring.
More general definitions of elements and element-wise operations
in contractions enable the extension of the formalism and algorithms to partially symmetric contractions~\cite{SD_ETHZ_2015}.
%, which are beyond the scope of this paper.

Our notation departs from standard notation conventions used in the tensor decomposition literature~\cite{doi:10.1137/07070111X}, due to our need to work with variable-order tensors and contractions over arbitrary sets of modes.
We leverage the fact that different modes of a symmetric tensor are indistinguishable to keep our notation as concise and descriptive as possible.
Some basic conventions we employ include denoting tensors (including vectors and matrices) in bold font and denoting vectors with lower-case letters (variable-order tensors are denoted as upper-case letters even when they can be vectors).
However, elements of tensors (and of vectors) are denoted in regular (non-bold) font.
%To distinguish labels from index notation for tensors and scalars, we place the label in parentheses, e.g., $\B{T^{(1)}}$, $s^{(A)}$.

\begin{definition}
%We denote a $d$-tuple of integers using vector notation as $\tpl{i}{d} \defeq (i_1,\ldots, i_d)$. 
We denote a $d$-tuple of integers as $\tpl{i}{d} \defeq (i_1,\ldots, i_d)$. 
\end{definition}
These tuples will be used as tensor indices, and each will
most often range from $1$ to $n$, so $\tpl{i}{d}\in\inti{1}{n}^d$. We concatenate tuples using the notation
$\tpl{i}{d}\tpl{j}{f}\defeq(i_1,\ldots, i_d,j_1,\ldots, j_f)$ for any $\tpl jf \in \inti 1n^f$.
We combine ranges similarly $\inti 1n^d\inti 1n^f{=}\inti 1n^d{\otimes} \inti 1n^f{=}\inti 1n^{d+f}$.
%We refer to the first $g<d$ elements of $\tpl{i}{d}$ as $\tpl{i}{g}$.
%When summing over a all tuples in $\inti{1}{n}^{d}$, rather than writing $d$ summations, we will write simply $\sum_\tpl{i}{d}=\sum_{i_1=1}^n\ldots \sum_{i_d=1}^n$, in which the range is $\tpl{i}{d}\in\inti{1}{n}^{d}$.
\begin{definition}
We refer to the space of increasing $d$-tuples with values between $1$ and $n$ as $\enti{1}{n}{d}=\{\tpl{i}{d} : i_1\leq\ldots \leq i_d,  \tpl{i}{d}\in\inti 1n^d\}$.
%which means \(\forall \tpl{i}{d}\in\enti{1}{n}{d}, i_1\leq \ldots \leq i_d\).
We also refer to the space of strictly increasing tuples as $\lnti{1}{n}{d}=\{\tpl{i}{d} : i_1<\ldots < i_d,  \tpl{i}{d}\in\inti 1n^d\}$.
%so \(\forall \tpl{i}{d}\in\lnti{1}{n}{d}, i_1< \ldots < i_d\). 
\end{definition}
The number of increasing $d$-tuples between $1$ and $n$ is given by $|\enti{1}{n}{d}|=\chchoose{n}{d}\defeq {n+d-1 \choose d}$.
The set of increasing tuples will be useful in our algorithms, as $\enti{1}{n}{d}$ enumerates the unique tensor entries of an order $d$ symmetric tensor (defined below).

%We denote tensors in bold font letters, but their elements in regular font.
\begin{definition}
A tensor $\C T$ with order $d$ and all dimensions $n$ is a collection (multiset), % of elements 
\[\C T\defeq \lt(\Ce T_\tpl{i}{d} : \tpl id \in \inti 1n^d\rt).\]
\end{definition}
%Our contraction formulation assumes all tensor elements are on the same algebraic ring $R$. % (although they could alternatively be  or floating point red in~\cite{SD_ETHZ_2015}).
%We refer to the order of the tensor as the number of indices in the tensor and the dimension as the range of each index.
We will usually consider tensors with all dimensions equal to $n$.
Given an order $d$ tensor $\C A$, we will refer to its elements using the notation $\Ce A_\tpl{i}{d} = \Ce A_{i_1,\ldots, i_d}$. 

\begin{definition}\label{def:nsctr}
For any $s,t,v\geq 0$ with $\omega\defeq s+t+v$, we denote a {\bf tensor contraction} over $v$ indices %between tensors, %express a symmetrized contraction between symmetric tensors,
between tensor $\C A$ of order $s+v$ and tensor $\C B$ of order $v+t$, 
into tensor $\C C$ of order $s+t$ each with all dimensions equal to $n$ as 
%via any operator $``\cdot" \in R_A\times R_B \rightarrow R_C$, as
\begin{align}
\C{C} =&\nsctr{\C A}{\C B}{s}{t}{v}  \defeq \forall \tpl{j}{s}\tpl{l}{t}\in\inti{1}{n}^{s+t}, \Ce C_{\tpl{j}{s}\tpl{l}{t}} = \sum_{\tpl{k}{v}\in \inti 1n^v} \Ce A_{\tpl{j}{s}\tpl{k}{v}}\cdot \Ce B_{\tpl{k}{v}\tpl{l}{t}}.\label{eq:defctr} 
\end{align}
\end{definition}
%We will focus on the case where $R=R_A=R_B=R_C$, so all the tensor elements
%are in some (possibly nonassociative) ring~\cite{schafer1966introduction} $(R,+,\cdot)$. 
%We will concentrate our analysis and proofs 
%on cases where at least two of $s,t,v$ are nonzero, since otherwise one of the tensors is a scalar and the problem is trivial. 
Throughout further contraction definitions and algorithms we will always denote $\omega\defeq s+t+v$ and assume $n\gg \omega$.
We employ the notation $\nsctr{\C A}{\C B}stv$, since when $v=0$, the operator $\otimes_0$ is equivalent to the tensor product, which is commonly denoted as $\otimes$.
\begin{definition}\label{def:tnslike}
We refer to contractions with exactly one of $s,t,v$ is zero, as {\it matrix-vector-like} and contractions with $s,t,v>0$ as {\it matrix-matrix-like}.
\end{definition}

We define symmetric tensors and symmetrized contractions by considering all possible permutations of their indices.
For this task, we introduce the following permutation notation.
\begin{definition}\label{def:perms}
Let $\perms{d}$ be the set of all possible $d$-dimensional permutation functions,
where each $\pi\in \perms{d}$ \rev{is described by a bijection 
$\hat{\pi} : \inti{1}{d} \to \inti{1}{d}$}, as 
$\pi(\tpl{i}{d})\defeq(i_{\hat{\pi}(1)},\ldots, i_{\hat{\pi}(d)})$, so
%The number of such functions is 
$|\perms{d}|=d!$. Denote the 
collection of all permutations of a tuple $\tpl{i}{d}$ as
$$\perms{}(\tpl{i}{d})\defeq \lt(\pi(\tpl{i}{d}) : \pi \in \perms{d}\rt).$$
\end{definition}
%We note that the collection of all permutations of a tuple with repeated values will 
%include identical tuples, e.g., $\perms{}((1,2,2))=((1,2,2),(1,2,2),(2,1,2),(2,1,2),(2,2,1),(2,2,1))$.
%Therefore, $|\perms{}(\tpl{i}{d})|=d!$ for any $\tpl{i}{d}$.

\begin{definition}\label{def:symmetric}
We say an order-$d$ tensor $\C{T}$ with all dimensions equal to $n$ is symmetric if 
\[\forall \tpl{i}{d}\in \inti{1}{n}^d,\tpl jd \in \perms{}(\tpl id),\quad \Ce T_\tpl{i}{d} = \Ce T_\tpl jd.\]
\end{definition}
According to Definition~\ref{def:symmetric}, scalars and vectors are symmetric tensors of order 0 and 1, respectively.
%We note that for any $d,f\geq 0$, the union of all concatenations of all possible permutations of the ordered subcollections produced by $\chi^d_f$ is
%equal to all possible permutations of the partitioned tuple,
%$$\perms{}(\tpl{k}{d+f})=\big[ \tpl{i'}{d}\tpl{j'}{f} : (\tpl{i}{d},\tpl{j}{f})\in \chi(\tpl{k}{d+f}), \quad \tpl{i'}{d}\in\Pi(\tpl{i}{d}),\quad \tpl{j'}{f}\in \Pi(\tpl{j}{f})\big].$$
\begin{definition}\label{def:fulsym}
For any $s,t,v\geq 0$, a {\bf symmetric contraction} 
is a contraction between symmetric tensors $\C A$ and $\C B$ into $\C C$, where the result is symmetrized, i.e.,
$\C C =  \syctr{\C A}{\C B}{s}{t}{v}$, implies
\begin{align}
&\forall \tpl{i}{s+t}\in\inti{1}{n}^{s+t}, 
 \Ce C_{\tpl{i}{s+t}}=\sum_{\tpl{j}{s}\tpl{l}{t} \in \perms{}(\tpl{i}{s+t})} \bigg(\sum_{\tpl{k}{v}\in \inti 1n^v} \Ce A_{\tpl{j}{s}\tpl{k}{v}}\cdot \Ce B_{\tpl{k}{v}\tpl{l}{t}}\bigg).\label{eq:symtrctr} 
\end{align}
\end{definition}
The resulting tensor $\C C$ satisfying \eqref{eq:symtrctr} is always symmetric. 
For ($s=1,t=0,v=1$), \eqref{eq:symtrctr} corresponds to the product of a symmetric matrix $\B A$ with a vector $\B b$,
$\syctr{\B A}{\B b}{1}{0}{1} \defeq \B A\B b$. 
For ($s=1,t=1,v=0$),% and commutative ``$\cdot$'', 
\eqref{eq:symtrctr} becomes the rank-two vector outer product of a column vector $\B a$ and a row vector
$\B b$, $\syctr{\B a}{\B b}{1}{1}{0}\defeq \B a \B b+{\B b}^\mathsf{T} {\B a}^\mathsf{T}$.
These two vector routines are members of the BLAS~\cite{lawson1979basic} and are building blocks
in a multitude of numerical routines.
%Our definition of tensors does not distinguish between row and column vectors, as our definition
%of contractions permits vectors to behave as either, e.g., $\syctr{\C a}{\C b}{0}{0}{1}$ is the inner product. 
For ($s=1,t=1,v=1$),
% and commutative ``$\cdot$'', 
\eqref{eq:symtrctr} becomes symmetrized multiplication of symmetric $n\times n$ matrices $\B A$ and $\B B$, 
%$\syctr{\C A}{\C B}{1}{1}{1}\defeq \C A \C B + {\C B}^\mathsf{T} {\C A}^\mathsf{T}$.
%
%\begin{definition}\label{def:fulsym}
%When $\C A$ and $\C B$ are symmetric, we call a symmetrized contraction between $\C A$ and $\C B$, $\syctr{\C A}{\C B}{s}{t}{v}$ a {\bf symmetric contraction}.
%\end{definition}
%An example of a a symmetric contraction is the symmetrized multiplication %in the case when $s=t=v=1$ and scalar multiplication, ``$\cdot$", is commutative, 
%of real symmetric $n$-by-$n$ matrices 
$\B{C}=\syctr{\B A}{\B B}{1}{1}{1}\defeq \B{A}\B{B}+\B{B} \B{A}$. 
Our definition of symmetric contractions can be extended to scenarios where the operands and/or the result are partially symmetric via nested tensors~\cite{SD_ETHZ_2015}. 
%, by defining formalism for nested contractions and algorithms
%We leave the derivation of communication lower bounds of such nested algorithms for partially symmetric contractions as future work.

While we will define symmetrization in contractions as summing over all possible permutations of the tensor indices (for any $\tpl{i}{d}$ the collection $\perms{}(\tpl{i}{d})$),
our algorithms will exploit the equivalence of many of these permutations.
As a result, they will need to sum over a set of partitions rather than a full set of permutations, as defined below.
\begin{definition}\label{def:part}
We define the {\bf disjoint partition} % FIXME: should be a collection and not a set I think.
$\prt{p}{q}(\tpl{k}{r})$ as the collection of all pairs of tuples of size $p$ and $q$,
which are disjoint subcollections of $\tpl{k}{r}$ and preserve the ordering of elements in $\tpl{k}{r}$.
We additionally define $\bprt{p}{q}(\tpl{k}{r})$ as the set of all unique pairs in $\prt{p}{q}(\tpl{k}{r})$.
\end{definition}

In other words, if $k_i$ and $k_j$ appear in the same tuple (partition) and $i<j$, then $k_i$ must appear before $k_j$. 
For example, the possible ordered partitions of $\tpl{k}{3}=(k_1,k_2,k_3)$ into
pairs of tuples of size one and two are the collection,
$$\prt 12(\tpl{k}{3})=\big[(k_1,(k_2,k_3)), (k_2,(k_1,k_3)), (k_3,(k_1,k_2))\big].$$
The collection $\prj{p}_q(\tpl{k}{r})$ can be constructed inductively~\cite{SD_ETHZ_2015}. 
The set $\bprt{p}{q}(\tpl{k}{r})$ will be used whenever we want to exclude equivalent pairs in $\prj{p}_q(\tpl{k}{r})$ (these exist only when $\tpl kr$ has repeating entries).

\begin{definition}
\label{def:prj}
We denote all possible ordered subcollections of tuple $\tpl{k}{d+f}\in\inti{1}{n}^{d+f}$ via the projection map $\prj{d}$,
\[\prj{d}(\tpl{k}{d+f})\defeq \lt(\tpl{i}{d} : (\tpl{i}{d},\tpl{j}{f})\in \prt df(\tpl{k}{d+f})\rt).\]
\end{definition}

In certain cases, our algorithms compute summations over groups of indices of symmetric tensors by summing only over the unique values and scaling by the following multiplicative factor. 
%In all such summations, the symmetry of the summed indices can be exploited to perform fewer operations, so we introduce the following special summation notation,
%which sums over only increasing indices and scales the result by an appropriate multiplicative factor.
\begin{definition}\label{def:symsum}
Let $\rho(\tpl{k}{v})\defeq v!/\prod_{i=1}^l m_i!$ where $m_i$ is the multiplicity of the $i$th of  $1\leq l \leq v$ unique values in $\tpl{k}{v}$.
% and then define
%\[\sysum_\tpl{k}{v}f(\tpl{k}{v})\defeq \sum_{\tpl{k}{v}\in\enti{1}{n}{v}} \rho(\tpl{k}{v}) f(\tpl{k}{v}),\]
%where $f$ is any function whose range is inside an Abelian group.
\end{definition}
The factor $\rho(\tpl{k}{v})$ is the number of unique permutations of $\tpl kv$, i.e., unique \rev{tuples} in the collection $\perms{}(\tpl kv)$ \rev{(if $\tpl kv$ has repeating values, some of its permutations will result in the same tuple).}

\section{Bilinear Algorithms for Tensor Contractions}
\label{sec:algs}
Tensor contractions ($\nsctr{\C A}{\C B}stv$) and symmetric tensor contractions ($\syctr{\C A}{\C B}stv$) produce a set of bilinear forms (partial sums of products) of the elements of $\C A$ and $\C B$.
We will define direct algorithms that compute these bilinear forms naively by computing all unique products of input elements and accumulating them to the specified partial sums.
These algorithms follow directly from the algebraic definitions of $\nsctr{\C A}{\C B}stv$ and $\syctr{\C A}{\C B}stv$ given in Section~\ref{sec:defs}. 
We then consider algorithms that compute a smaller set of products of linear combinations of input elements and obtain the specified partial sums as linear combinations of these.
We then provide specifications of them as {\it bilinear algorithms}~\cite{pan1984can}, representing each algorithm as a 3-tuple of sparse matrices.

\subsection{Nonsymmetric Contraction Algorithm}
\label{subsec:dealg}

Nonsymmetric tensor contractions are reducible to matrix multiplication.
So, we first consider the trivial algorithm which contracts nonsymmetric tensors $\C A$ and $\C B$ by evaluating the products in \eqref{eq:defctr}, which corresponds to standard matrix multiplication.
\begin{algorithmdef}[${\nnsalg stv}$]\label{alg:nsctr}
For any tensor contraction $\C C=\nsctr{\C A}{\C B}stv$ %\cdot(\B A,\B B)$ with any operator $``\cdot" \in R_A\times R_B \rightarrow R_C$,
we define $\C C =\nsalg{\C A}{\C B}stv$\ to evaluate the multiplications explicitly described by \eqref{eq:defctr} in Definition~\ref{def:nsctr}.
\end{algorithmdef}
%\fixed{Edgar}{Explain how this algorithm maps to matrix multiplication.}
\rev{Algorithm~\ref{alg:nsctr} is equivalent to a multiplication of a matrix $\BB A$ with dimensions $n^s\times n^v$ by a matrix $\BB B$ with dimensions $n^v\times n^t$ yielding a matrix $\BB C$ with dimensions $n^s\times n^t$.
These matrices are referred to as {\it unfoldings} of $\C A$, $\C B$, and $\C C$, respectively.
An unfolding of a tensor is a lower order tensor whose modes combine subsequences of modes in the higher-order tensor, while retaining its overall size and set of elements~\cite{doi:10.1137/07070111X}.
Specifically, an element $\bar a_{\bar{j}\bar{k}}$ of $\BB A$ is equal to $\Ce A_{\tpl js\tpl kv}$, where the row index $\bar{j}$ maps to a unique $\tpl{j}{s}\in\inti{1}{n}^{s}$ and the column index $\bar{k}$ maps to a unique $\tpl{k}{v}\in\inti{1}{n}^{v}$.
$\C B$ and $\C C$ are defined similarly using the same mappings of row/column indices to tensor index tuples.
}
%
%, where each row corresponds to $\tpl{j}{s}\in\inti{1}{n}^s$ and
%each column corresponds to $\tpl{k}{v}\in\inti{1}{n}^v$, 
%, where each row corresponds to $\tpl{k}{v}\in\inti{1}{n}^v$ and each column

When $s,t,v>0$, we could alternatively employ a different matrix multiplication algorithm to compute \eqref{eq:defctr} (e.g., Strassen's algorithm~\cite{Strassen_1969}).
In this paper, we will not consider such fast matrix multiplication algorithms, focusing instead on algorithms that exploit symmetry.
\subsection{Direct Evaluation Algorithm for Symmetric Contractions}

The nonsymmetric algorithm may be used to compute symmetric contractions with the additional
step of symmetrization of the result of the multiplication between $\C A$ and $\C B$. However,
when $\C A$ and $\C B$ are symmetric, many of the scalar
multiplications (products) in \eqref{eq:symtrctr} are equivalent. The following algorithm evaluates \rev{$\C C = \syctr{\C A}{\C B}stv$} % $\Psi_\odot^{(s,t,v)}$
by computing only the unique multiplications and scaling them appropriately.

In particular, since $\C{C}$ is symmetric, it is no longer 
necessary to compute all possible orderings of the indices
$\tpl{i}{s+t}\in\inti{1}{n}^{s+t}$ in $\syctr{\C A}{\C B}stv$, but only those in increasing order.
These tuples are given by $\tpl{i}{s+t}\in\enti{1}{n}{s+t}$
%(meaning $i_1\leq \ldots \leq i_{s+t}$), 
and index into all unique values of $\C C$.
Further, permutations of the $\tpl kv$ \rev{tuple} result in 
equivalent scalar multiplications due to symmetry of $\C{A}$ and of $\C{B}$. 
In the following algorithm, we rewrite \eqref{eq:symtrctr}  
to sum over only the ordered sets of these indices and scale them
appropriately.
%{\bf Algorithm $\mathbf{\systalg{\C A}{\C B}stv}$ for Symmetric $\C{A}$ and $\C{B}$ :}
\begin{algorithmdef}[$\nsystalg stv$]\label{alg:systsy}
For any symmetric contraction $\C C=\syctr{\C A}{\C B}stv$ compute %with any operator $``\cdot" \in R_A\times R_B \rightarrow R_C$,
%(see Definitions~\ref{def:symctr} and~\ref{def:fulsym}) compute
\begin{align}
&\C C = \systalg{\C A}{\C B}stv 
  \defeq \forall \tpl{i}{s+t}\in\enti{1}{n}{s+t}, \nonumber \\
&\Ce C_{\tpl{i}{s+t}}= s!t! 
\sum_{(\tpl js,\tpl lt)\in \prt st(\tpl{i}{s+t})} 
\bigg(\sum_{\tpl kv\in \enti 1nv} \rho(\tpl{k}{v})\Ce A_{\tpl js\tpl kv}\cdot \Ce B_{\tpl kv\tpl lt}\bigg),\label{eq:symctr}
\end{align}
where $\rho(\tpl kv)$ is given in Definition~\ref{def:symsum}
\end{algorithmdef}
\fixed{Torsten}{Goes very fast. Resolution (Edgar): added details about folding into matrices.}
%$\B C = \systalg{\C A}{\C B}stv$ employs the summation notation from Definition~\ref{def:symsum}.
% and assumes that elements of $R$ can be scaled by positive integers (this can always be done with repeated additions).
The algorithm $\nsystalg stv$ is algebraically equivalent to~\eqref{eq:symtrctr} and is numerically stable~\cite{solomonikfast}.

%Modulo the scaling by $s!t!$ and $\rho(\tpl{k}{v})$, % (implicit in $\sysum_\tpl{k}{v}$), 
\rev{
The inner summation of Algorithm~\ref{alg:systsy} is equivalent to multiplication of a matrix $\BB A$ with dimensions $\chchoose ns\times \chchoose nv$ 
by a matrix $\BB B$ with dimensions $\chchoose nv\times \chchoose nt$, yielding a matrix $\BB C$ with dimensions $\chchoose ns\times \chchoose nt$. 
An element $\bar a_{\bar{j}\bar{k}}$ of $\BB A$ is equal to $\Ce A_{\tpl js\tpl kv}$, where the row index $\bar{j}$ maps to a unique $\tpl{j}{s}\in\enti{1}{n}{s}$ and the column index $\bar{k}$ maps to a unique $\tpl{k}{v}\in\enti{1}{n}{v}$.
%Each row of $\BB A$ corresponds to some $\tpl{j}{s}\in\enti{1}{n}{s}$ and each column corresponds to $\tpl{k}{v}\in\enti{1}{n}{v}$).
%Each row of $\BB B$ corresponds to some $\tpl{k}{v}\in\enti{1}{n}{v}$ and each column corresponds to $\tpl{l}{t}\in \enti{1}{n}{t}$.
}
The matrix $\BB C$ \rev{corresponds to} a partially symmetric tensor
of order $s+t$ (symmetric in the permutation within the first $s$ indices and within the last $t$ indices) and can be further symmetrized to obtain $\C C =\syctr{\C A}{\C B}stv$,
\[\forall \tpl i{s{+}t}\in\enti{1}{n}{s+t}, \Ce C_{\tpl{i}{s+t}}= s!t!\sum_{(\tpl js,\tpl lt)\in \prt st(\tpl{i}{s+t})} \Ce{\bar{C}}_{\tpl{j}{s}\tpl{l}{t}}.\]
The scaling by $\rho(\tpl{k}{v})$ can be applied to the elements of $\BB A$ or of $\BB B$ \rev{to enable direct reduction to a matrix-matrix product. }
% (its preferable to apply the scaling to the tensor which has smaller size).

\subsection{Symmetry Preserving Tensor Contraction Algorithm}

The symmetry preserving tensor contraction algorithm below computes symmetric contractions with fewer multiplications and in some cases fewer total operations than
the direct evaluation algorithm~\cite{solomonikfast}.
It has a number of applications in both matrix computations and high-order coupled cluster~\cite{cizek66} tensor contractions.
The algorithm requires the computation of a few intermediate tensors, but for brevity we only give the formula for the highest order intermediate tensor (order $\omega=s+t+v$) computed by the algorithm (the rest may be computed using $O(n^{\omega-1})$ multiplications), which suffices for our lower bound derivations.

\rev{
\begin{algorithmdef}[$\nsyfsalg stv$]\label{alg:syfs}
For any symmetric contraction $\C C=\syctr{\C A}{\C B}stv$, compute 
\begin{align}
\forall \tpl{i}{\omega}\in\enti{1}{n}{\omega}, \quad\quad 
\Ce{\hat{A}}_{\tpl{i}{\omega}} &= \sum_{\tpl{j}{ s+v } \in \prj{s+v}( \tpl{i}{\omega})}\Ce A_\tpl{j}{s+v}, \label{eq:fscAacc} \\
\Ce{\hat{B}}_{\tpl{i}{\omega}} &= \sum_{\tpl{l}{ v+t } \in \prj{v+t}( \tpl{i}{\omega})} \Ce B_\tpl{l}{v+t}, \label{eq:fscBacc} \\
\Ce{\hat{Z}}_{\tpl{i}{\omega}} &= \Ce{\hat{A}}_{\tpl i\omega} \cdot \Ce{\hat{B}}_{\tpl{i}{\omega}}, \label{eq:fscmul} \\
\forall \tpl{h}{s+t}\in\enti{1}{n}{s+t}, \quad\quad 
 \Ce Z_\tpl{h}{s+t} &= \sum_{\tpl{i}{\omega}\in\enti{1}{n}{\omega}, \tpl{h}{ s+t } \in \prj{s+t}( \tpl{i}{\omega})}\Ce{\hat{Z}}_{\tpl{i}{\omega}}.  \label{eq:fscZacc}
\end{align}
%where $\rho(\tpl kv)$ is given in Definition~\ref{def:symsum}.
Then compute $\C C =\syfsalg{\C A}{\C B}stv = s!t!(\C Z-\C V - \C W)$, where formulas for computing $\C V$ and $\C W$ are given in ~\cite{solomonikfast}.
%For any symmetric contraction $\C C=\syctr{\C A}{\C B}stv$, compute an order $\omega$ symmetric tensor $\B{\hat{Z}}$ whose unique elements are 
%each computed by a single scalar multiplication: 
%\begin{align}
%\forall \tpl{i}{\omega}\in\enti{1}{n}{\omega},
%\hat{Z}_\tpl{i}{\omega} = \bigg(\sum_{\tpl{j}{ s+v } \in \prj{s+v}( \tpl{i}{\omega})} A_{\tpl{j}{ s+v }}\bigg) \cdot
%      \bigg(\sum_{\tpl{l}{ v+t } \in \prj{v+t}( \tpl{i}{\omega})} B_{\tpl{l}{ v+t }}\bigg). 
%       \label{eq:fschatZ}
%\end{align}
%Accumulate $\BH Z$ into a tensor $\C Z$, 
%\begin{align}
%\forall \tpl{h}{s+t}\in\enti{1}{n}{s+t}, 
%Z_{\tpl{h}{ s+t }} = & \sum_{\tpl{k}{v}\in \enti 1nv}\rho(\tpl kv) \hat{Z}_{\tpl{h}{s+t}\tpl{k}{v}}, \label{eq:fscZ}
%\end{align}
%where $\rho(\tpl kv)$ is given in Definition~\ref{def:symsum}.
%Compute $\C C = s!t!(\C Z-\C V - \C W)$, where formulas for computing $\C V$ and $\C W$ are given by \cite{SD_ETHZ_2015}.
%%by Solomonik and Demmel~\cite{SD_ETHZ_2015}.
\end{algorithmdef}
}

The $\C Z$ tensor contains all terms needed by $\B C=\syctr{\C A}{\C B}stv$ as well as some extra terms which are independently computed as tensors $\C V$ and $\C W$ then subtracted out from $\C Z$.
The computation of $\C V$ and $\C W$ can always be done via a low order number of multiplications, but sometimes requires a constant factor more additions than those needed to compute $\C Z$.

%The most basic instance of the symmetry preserving algorithm is symmetric matrix vector multiplication, $\B c=\syctr Ab101 \defeq \C A \B b$, for which the algorithm computes $\C V$ but not $\C W$,
%\begin{align*}
%&\forall i\in\inti{1}{n},  j\in\inti{i}{n},  \hat{A}_{ij}=A_{ij},  \hat{B}_{ij} = b_i+b_j,  \hat{Z}_{ij} = \hat{A}_{ij}\cdot \hat{B}_{ij}, \notag \\
%& Z_{i} = \sum_{k=1}^i \hat{Z}_{ki} + \sum_{k=i+1}^n \hat{Z}_{ik},
% A^{(1)}_{i} = \sum_{k=1}^n A_{ik}, 
% V_{i} = A^{(1)}_{i}\cdot b_i, 
% c_i = Z_i - V_i. \label{eq:symv2} 
%\end{align*}
%The advantage of $\B c=\syctr Ab101$ is that it requires only $(1/2)n^2$ multiplications to leading order instead of the usual $n^2$.
%The tensor $\C W$ needs to be computed only when $s$ and $t$ are greater than zero, the most basic example of which is the symmetrized vector outer product~\cite{SD_ETHZ_2015}.

%The correctness proof, numerical stability proof, numerical tests, computation cost analysis, adaptations from symmetric to partially symmetric, antisymmetric, and Hermitian cases, as well as applications are given in \cite{SD_ETHZ_2015}.
The algorithm is numerically stable and can be extended to the antisymmetric/Hermitian cases~\cite{SD_ETHZ_2015}.
The algorithm generalizes tensor formula reorganizations previously used to lower cost of methods for electronic structure calculations~\cite{doi:10.1080/00268970500131140}.
The symmetry preserving algorithm 
\rev{does not reduce to} matrix multiplication unlike $\nnsalg stv$ and $\nsystalg stv$, which makes its communication cost analysis nontrivial.
%The symmetry preserving algorithm may leverage fast matrix multiplication algorithms.

%\begin{theorem}
%The computation cost of the algorithm $\systalg{\C A}{\C B}stv$ for symmetric
%$\C{A}$ and $\C{B}$ (ignoring the cost of applying constant scaling factors) is %then given to leading order by
%\begin{align*}
%F^\Psi(\varrho,n,s,t,v)&=F^{\mathrm{MM}}\lt(\varrho,\chchoose{n}{s},\chchoose{n}{t},\chchoose{n}{v}\rt) + \nu^C_\varrho \lt[\chchoose{n}{s}\chchoose{n}{t} -\chchoose{n}{s+t}\rt] \\
%&= (\nu^C_\varrho+\mu_\varrho) \chchoose{n}{s}\chchoose{n}{t}\chchoose{n}{v} -\nu^C_\varrho \chchoose{n}{s+t} \approx (\nu^C_\varrho+\mu_\varrho) \frac{n^{\omega}}{s!t!v!}.
%\end{align*}
%% + \nu^C_\varrho \frac{n^{s+t}}{s!t!}\]% + O(n^{\omega-1}(\nu_\varrho+\mu_\varrho)).\] 
%%$$F(\varrho,n,s,t,v)={n\choose s}{n\choose t}{n\choose v}\cdot(\nu_\varrho+\mu_\varrho) +
%%{s+t\choose s}{n\choose s}{n\choose t}\cdot \nu_\varrho$$
%\end{theorem}
%\begin{proof}
%\[F^\Psi(\varrho,n,s,t,v)=F^{\mathrm{MM}}\lt(\varrho,\chchoose{n}{s},\chchoose{n}{t},\chchoose{n}{v}\rt) +  \nu^C_\varrho\lt[ \chchoose{n}{s}\chchoose{n}{t}-\chchoose{n}{s+t}\rt].\]
%\end{proof}

\subsection{Bilinear Algorithms}

%All of the above algorithms are generalized by 
Bilinear algorithms \cite{pan1984can} provide a unified algebraic representation for all algorithms above.
%Our definition of bilinear algorithms follows that of .
%, except we employ some index sets of a higher order.
\begin{algorithmdef}[$\bilalg F$]\label{alg:bil}
Given input vectors $\B a$, $\B b$, a bilinear algorithm, defined by matrices $\mathcal{F} = (\bFA, \bFB, \bFC)$ (with values in $\mathbb{R}$), computes output vector $\B c$ via
\[\B{{c}}=\bilalg F(\B a, \B b)\defeq \bFC[({\bFA}{}^\mathsf{T}\B{{a}})\odot({\bFB}{}^\mathsf{T}\B{{b}})],\]
where $\odot$ is the Hadamard (pointwise) product. 
\end{algorithmdef}
We say that a bilinear algorithm is {\it irreducible} if each of the matrices in $\mathcal{F}$ has full row rank.
For typical bilinear algorithms the three matrices in $\mathcal{F}$ are very sparse.
For instance, the classical matrix multiplication algorithm, has only one (unit) entry in each column of the three matrices.
The number of rows in $\bFA$, $\bFB$, $\bFC$ is equal to the dimension of $\B a$, $\B b$, $\B c$, respectively.
We refer to this 3-tuple of dimensions of $\B a$, $\B b$, $\B c$ as $\dim(\bilalg F) = (\dim(\B a), \dim(\B b), \dim(\B c))$.
There is an equal number of columns in the three matrices $\bFA$, $\bFB$, $\bFC$, the number of which is referred to as the {\it rank} of the bilinear algorithm, which we denote by $\rank(\bilalg F)$.
The rank of the bilinear algorithm defines the number of scalar multiplications in the Hadamard product, which corresponds to the number of products the algorithm computes.
Each column of the sparse matrices in $\bFA$ and $\bFB$ defines the linear combination of inputs (elements of $\B a$ or of $\B b$) that contribute to a particular product.
Each row of $\bFC$ defines how each output element of $\B c$ is computed as a linear combination of products.

We now define a representation for {\it bilinear tensor algorithms}.
Each bilinear tensor algorithm corresponds to a bilinear algorithm in the above representation.
The utility of the tensor representation is to provide notation to succinctly express algorithms for tensor contractions as bilinear algorithms.
%This representation is no more general than the bilinear algorithm formulation above, but it allows
%Our form is not able to express more computations that than of \cite{pan1984can}, but it reduces to the desired special cases of tensor contraction algorithms more succinctly.
\begin{algorithmdef}[$\nbilalg stvDK$]\label{alg:nbil}
Given tensors $\C A$, $\C B$, the bilinear tensor algorithm computes tensor $\C C$.
\rev{The three tensors and the set of bilinear products are associated with domains defined by} $\mathcal{D}\defeq (D^{(A)},D^{(B)},D^{(C)},D^{(M)})$ where $D^{(A)}\subseteq \inti{1}{n}^{s+v}$, $D^{(B)}\subseteq \inti{1}{n}^{v+t}$, $D^{(C)}\subseteq \inti{1}{n}^{s+t}$, and $D^{(M)}\subseteq {\inti 1n}^\omega$ (as before $\omega=s+t+v$).
The algorithm computes the products specified by $\mathcal{K}\defeq (\C{K}^{(A)}, \C{K}^{(B)}, \C{K}^{(C)})$ where $\C{K}^{(A)}$, $\C{K}^{(B)}$, and $\C{K}^{(C)} $ are (sparse) tensors (with values in $\mathbb{R}$) of orders
$\inti{1}{n}^{s+v+\omega}$, 
$\inti{1}{n}^{v+t+\omega}$, and 
$\inti{1}{n}^{s+t+\omega}$, respectively. 
The bilinear algorithm computes the tensor $\C{C}$ by evaluating the following set of products (corresponding to intermediate tensor $\C T$),
\[ \forall \tpl i{\omega} \in D^{(M)}, \quad \Ce T_{\tpl i\omega} = \Bigg(\sum_{\tpl{j}{s+v}\in D^{(A)}} \Ce K^{(A)}_{\tpl{j}{s+v}\tpl{i}{\omega}}\Ce A_\tpl{j}{s+v}\Bigg)\cdot\Bigg(\sum_{\tpl{l}{v+t}\in D^{(B)}} \Ce K^{(B)}_{\tpl{l}{v+t}\tpl{i}{\omega}}\Ce B_\tpl{l}{v+t}\Bigg),\]
then it accumulates these products to compute $\C C$,
\[ \forall \tpl h{s+t}  \in D^{(C)}, \quad \Ce C_{\tpl{h}{s+t}} =  \sum_{\tpl{i}{\omega}\in D^{(M)}} \Ce K^{(C)}_{\tpl{h}{s+t}\tpl{i}{\omega}} \Ce T_{\tpl{i}{\omega}}.\]
\end{algorithmdef}
To obtain a bilinear algorithm $\bilalg F$ from $\nbilalg stvDK$ it suffices to enumerate the domains in $\mathcal{D}$, mapping the corresponding four index sets \rev{(four tuples)} to four indices.
\rev{We refer to the process of mapping the data of a tensor $\C T$ from index domain $D^{(T)}$ to a vector of size $|D^{(T)}|$ as a {\it restricted unfolding} of $\C T$.
%This terminology is often used to describe transformation of nonsymmetric tensors into a lower order tensor or a vectors~\cite{doi:10.1137/07070111X}.}
The resulting bilinear algorithms would operate on vectors $\B a$ and $\B b$ which are unfoldings of the 
the input tensors $\C A$ and $\C B$ restricted to domains $D^{(A)}$ and $D^{(B)}$, respectively.
The matrices $\bFA$, $\bFB$, and $\bFC$ in the bilinear algorithm correspond to restricted unfoldings of the tensors $\C{K}^{(A)}$, $\C{K}^{(B)}$, and $\C{K}^{(C)} $  used in the bilinear tensor algorithm.}
% into matrix representations, $\mathcal{F}$.
\begin{definition}
%A bilinear algorithm $\bilalg F$ is the canonical representation of 
For any bilinear tensor algorithm $\nbilalg stvDK$, we denote its unique canonical bilinear algorithm as
\[\bilalg F = \getmat{\nbilalg stvDK},\]
\rev{where the restricted unfolding of $\mathcal{K}$ into $\mathcal{F}$ is defined by a lexicographical ordering of the domains in $\mathcal{D}$.}
%For any bilinear tensor algorithm, $\nbilalg stvDK$, we denote the unique canonical bilinear algorithm as
\end{definition}
The canonical bilinear algorithm will satisfy $\dim(\getmat{\nbilalg stvDK}) = (|D^{(A)}|, |D^{(B)}|, |D^{(C)}|)$ and $\rank(\getmat{\nbilalg stvDK}) = |D^{(M)}|$.

%It is implied that Algorithm~\ref{alg:bil} will evaluate only non-trivial summations, i.e., only sum terms for which the scaling given by one of the tensors in $\mathcal{F}$ is non-zero.
%The number of bilinear forms computed is referred to as the rank of the bilinear algorithm, $\rank(\nbilalg stvDF)=|D^{(M)}|$ .
%While each bilinear form may require many additions, the rank controls the complexity of the bilinear algorithm when applying ``$\cdot$'' dominates in cost the additions.
%% (applications of the various $\mathcal F$ integer tensors).
%We also assume that the bilinear algorithm makes use of all values in the input domain, i.e., 
%\(\forall \tpl{j}{s+v}\in D^{(A)}, \exists \tpl i\omega \ \text{such that} \ F^{(A)}_{\tpl{j}{s+v}\tpl{i}{\omega}}\neq 0,\)
%and similarly for $D^{(B)}$.
%Additionally, we assume that all computed values in the output-domain $D^{(C)}$ are non-zero and non-trivial.

The tensor domains of unique values are tailored for the symmetric tensor contraction algorithms (e.g., these domains are increasing integer sets for symmetric tensors).
In Table~\ref{tab:bil_cst} and Table~\ref{tab:bil}, we give the particular values of $\mathcal D$ and $\mathcal K$, respectively, that yield the three tensor contraction algorithms considered in this paper.
The bilinear rank of the direct evaluation algorithm is $\rank(\nsystalg stv)=\chchoose{n}{s}\chchoose{n}{t}\chchoose{n}{v}$, while the bilinear rank of the symmetry preserving algorithm is $\rank(\nsyfsalg stv)=\chchoose{n}{\omega}$ (excluding consideration of $\C V$ and $\C W$, which can be computed with bilinear rank $O(n^{\omega-1})$).

\begin{table*}[h]
\centering
%\begin{tabular}{l|ccc|c|c|c|c}
%\tbl{
\caption{
Bilinear tensor algorithm domain and ranks for the three tensor contraction algorithms.}
\small
\begin{tabular}{l|c|c|c|c|c}
%\setlength\tabcolsep{4pt}
%{\small $(s,t,v)$}
%Alg. 
$\nbilalg stvDK$ & $D^{(A)}$ & $D^{(B)}$ & $D^{(C)}$ & $D^{(M)}$ & $\rank(\nbilalg stvDK)$  \\[2ex]
\hline %&&&&&&& \\ [-1.75ex] 
%$\Upsilon$ &
$\nnsalg stv$  & ${\inti 1n}^{s+v}$ & ${\inti 1n}^{v+t}$ & ${\inti 1n}^{s+t}$ &  ${\inti 1n}^\omega$ & $n^\omega$ \\[2ex]
\hline
$\nsystalg stv$ & $\enti 1n{s+v}$ & $\enti 1n{v+t}$ & $\enti 1n{s+t}$ & $\enti 1ns \enti 1nt\enti 1nv$ & $\chchoose ns\chchoose nt\chchoose nv$  \\[2ex]
\hline
$\nsyfsalg stv$ & $\enti 1n{s+v}$ & $\enti 1n{v+t}$ & $\enti 1n{s+t}$ & $\enti 1n\omega$ & $\chchoose n\omega$ 
\end{tabular}
\label{tab:bil_cst}
\end{table*}

\begin{table*}[h]
\centering
%\begin{tabular}{l|ccc|c|c|c|c}
%\tbl{
\caption{
Representations of three tensor contraction algorithms as bilinear tensor algorithms, $\nbilalg stvDK$. 
\rev{Only nonzeros of $\C{K}^{(A)}$, $\C{K}^{(B)}$, $\C{K}^{(C)}$ are specified.}
%The last three columns give the non-zero fill structure of the tensors in $\mathcal{K}$.
To interpret $\bprt{}{}{}$ and $\rho$\rev{,} see Definition~\ref{def:part} and Definition~\ref{def:symsum} respectively. The new scaling factors given here for $\nsyfsalg stv$ arise due to multiple accumulations of tensor elements with repeating indices done within Algorithm~\ref{alg:syfs}.}
\scriptsize
\begin{tabular}{l|c|c|c}
%\setlength\tabcolsep{4pt}
%{\small $(s,t,v)$}
%Alg. 
$\nbilalg stvDK$  & $\C{K}^{(A)}$ & $\C{K}^{(B)}$ & $\C{K}^{(C)}$ \\[2ex]
\hline %&&&&&&& \\ [-1.75ex] 
$\nnsalg stv$ 
& \multicolumn{3}{c}{
\rev{$\forall \tpl js {\in} {\inti 1n}^s, \tpl lt {\in} {\inti 1n}^t, \tpl kv {\in} {\inti 1n}^v$,}}  \\[2ex]
& $k^{(A)}_{\tpl{j}{s}\tpl{k}{v}\tpl{j}{s}\tpl{l}{t}\tpl{k}{v}}{=}1$ &
  $k^{(B)}_{\tpl{k}{v}\tpl{l}{t}\tpl{j}{s}\tpl{l}{t}\tpl{k}{v}}{=}1$ &
  $k^{(C)}_{\tpl{j}{s}\tpl{l}{t}\tpl{j}{s}\tpl{l}{t}\tpl{k}{v}}{=}1$ \\[2ex]
\hline
%$\Psi_\circ$ & 
$\nsystalg stv$ 
%& \multicolumn{3}{c}{
%$\tpl lt {\in} \enti 1nt, \tpl kv {\in} \enti 1nv,$  
%} \\[2ex]
& $\forall \tpl lt {\in} \enti 1nt, \tpl g{s+v}{\in}\enti 1n{s+v},$ 
& $\forall \tpl js {\in} \enti 1ns, \tpl h{v+t}{\in}\enti 1n{v+t},$
& $\forall \tpl kv {\in} \enti 1nv, \tpl i{s+t}{\in}\enti 1n{s+t},$
\\[2ex]
%& $(\tpl js, \tpl kv){\in}\chi^s_v(\tpl g{s+v}),$ 
%& $(\tpl kv, \tpl lt){\in}\chi^v_t(\tpl h{v+t}),$
%& $(\tpl js, \tpl lt){\in}\chi^s_t(\tpl i{s+t}),$
%\\[2ex]
%& $\exists \tpl g{s+v}{\in}\enti 1n{s+v}, (\tpl js, \tpl kv){\in}\chi^s_v(\tpl g{s+v}),$ 
%& $\exists \tpl h{v+t}{\in}\enti 1n{v+t}, (\tpl kv, \tpl lt){\in}\chi^v_t(\tpl h{v+t}),$
%& $\exists \tpl i{s+t}{\in}\enti 1n{s+t}, (\tpl js, \tpl lt){\in}\chi^s_t(\tpl i{s+t}),$
%\\[2ex]
& $(\tpl js, \tpl kv){\in}\chi^s_v(\tpl g{s+v}),$ 
  $\Ce K^{(A)}_{\tpl{g}{s+v}\tpl{j}{s}\tpl{l}{t}\tpl{k}{v}}{=}1$ 
& $(\tpl kv, \tpl lt){\in}\chi^v_t(\tpl h{v+t}),$
  $\Ce K^{(B)}_{\tpl{h}{v+t}\tpl{j}{s}\tpl{l}{t}\tpl{k}{v}}{=}1$ 
& $(\tpl js, \tpl lt){\in}\chi^s_t(\tpl i{s+t}),$
  $\Ce K^{(C)}_{\tpl{i}{s+t}\tpl{j}{s}\tpl{l}{t}\tpl{k}{v}}{=}s!t!\rho(\tpl kv)$ \\[2ex]
%\\[2ex]
%&
%\hline &&&&&&& \\ [-1.75ex] 
\hline 
%$\Psi_\circ$ & 
$\nsyfsalg stv$ & 
\multicolumn{3}{c}{$\forall \tpl i\omega {\in} \enti 1n\omega,$} \\[2ex]
& $(\tpl j{s+v},\tpl{a}{t}){\in} \bprt{s+v}t(\tpl{i}{\omega}),$
  $\Ce K^{(A)}_{\tpl{j}{s+v}\tpl{i}{\omega}}{=}\frac{t!}{\rho(\tpl at)}$
& $(\tpl l{v+t},\tpl{b}{s}){\in} \bprt{v+t}s(\tpl{i}{\omega}),$
  $\Ce K^{(B)}_{\tpl{l}{v+t}\tpl{i}{\omega}}{=}\frac{s!}{\rho(\tpl bs)}$ 
& $(\tpl h{s+t},\tpl{c}{v}){\in} \bprt{s+t}v(\tpl{i}{\omega}),$ 
  $\Ce K^{(C)}_{\tpl{h}{s+t}\tpl{i}{\omega}}{=}s!t!\rho(\tpl cv)$
\end{tabular}
%}
\label{tab:bil}
\end{table*}

We will study the computation of subsets of products \rev{(bilinear forms) of a} bilinear algorithms.
To achieve this, we define a 
\rev{{\it restriction} of a bilinear algorithm by taking subsets of columns in the matrix encoding.}
\begin{definition}
We say a bilinear algorithm $\bilalg G$ is a restriction of bilinear algorithm $\bilalg F$ (denoted $\bilalg G \subseteq \bilalg F$) if there is a projection matrix $\B P\in \{0,1\}^{\rank(\bilalg F)\times\rank(\bilalg G)}$ (with a single unit entry per column and at most one entry per row), such that
\[\B{G}^{(A)}= \B{F}^{(A)}\B P, \quad \quad
\B{G}^{(B)}= \B{F}^{(B)}\B P, \quad \quad
\B{G}^{(C)}= \B{F}^{(C)}\B P.\]
For tensor bilinear algorithms, $\nbilalg stvHR \subseteq \nbilalg stvDF$, when 
$\getmat{\nbilalg stvHR} \subseteq \getmat{\nbilalg stvDF}$.
\end{definition}

\section{Model of Execution and Costs}
\label{sec:model}
We present lower bounds for communication between main memory and cache on a sequential computer as well as communication between computers in the parallel setting.
%We derive lower bounds for all the nonsymmetric and fully symmetric tensor contraction algorithms given in Section~\ref{sec:algs}. 
A communication cost is associated with a schedule of an execution DAG of an algorithm.
%, with the algorithm execution 
%We consider the directed acyclic graph (DAG) associated with an execution of an algorithm, and define measures of communication cost for a schedule of such an execution DAG.
In this section, we formally define execution DAGs for bilinear algorithms as well as how these executions are scheduled on sequential and parallel machine models.
The cache size and communication costs are implicitly parameterized by the tensor element size (all elements are assumed to be in the same ring $R$ and represented using a constant number of bits).
%%Our lower bounds are parameterized for the machine models discussed in Section~\ref{sec:model}, with $\cs$ being the cache size for a sequential machine and $M$ being the memory per processor (both in terms of tensor elements).

%\fixme{Edgar}{Move Loomis-Whitney theorem here and make intro more high level.}

%These communication lower bounds assume that the scalar multiplications defined in each algorithm in Section~\ref{sec:algs} are all performed once (no recomputation).
%Further, they assume that the summations are done in some ordering.
\subsection{Execution Model for Bilinear Algorithms}
\label{sec:exe_model}

Our definition of a bilinear algorithm allows for freedom in the implementation of the algorithm.
In particular, it specifies only which linear combinations need to be computed, and not through what intermediates they are computed or in which order.
We represent a choice of intermediates via a DAG.
\begin{definition}
\label{def:exeDAG}
An {\bf execution DAG} of a bilinear algorithm $\bilalg F$, is a directed acyclic graph $G=(V,E)$.
Each vertex in $V$ has either in-degree $0$, in which case it corresponds to an input value, or in-degree $2$, in which case it corresponds to the result of a (weighted) addition or a multiplication.
Further, $V$ may be partitioned disjointly as 
\[V=V^{(A)}\cup V^{(B)} \cup V^{(C)}.\]
The vertices in $V^{(A)}$ and $V^{(B)}$ correspond to linear combinations of $\C A$ and $\C B$ respectively, while the vertices in $V^{(C)}$ correspond to all computed \rev{products and linear combinations of products}.
Vertices corresponding to inputs are denoted by $V^{(I)}\subseteq V^{(A)}\cup V^{(B)}$.
Computed products (results of multiplication) and outputs are contained in $V^{(C)}$, we denote the computed products as $V^{(M)}\subseteq V^{(C)}$.
%We define the sets $V^{(A)}\in W^{(A)}$ and $V^{(B)}\in W^{(B)}$ to be the sets of vertices corresponding to inputs (all with in-degree 0). 
%Further, the set of vertices containing the outputs, $V^{(C)}$ is contained in $W^{(C)}$.
%So, we have $\dim(\bilalg F) = (|V^{(A)}|,|V^{(B)}|,|V^{(C)}|)$ and
We note that \(\rank(\bilalg F) = |V^{(M)}|.\)

%For $X\in\{A,B,M\}$, $V^{(X)}$ has exactly one vertex for each member of domain $D^{(X)}$
%and for $Y\in\{A,B,M\}$, the sets $W^{(Y)}$ contain all other intermediate (non-input and non-output) values computed as linear combinations of elements in domain $D^{(Y)}$.
The set of edges $E$ may be partitioned disjointly as
\[E=E^{(A)} \cup E^{(B)} \cup E^{(AC)} \cup E^{(BC)} \cup E^{(C)}.\]
%where
The parts $E^{(A)} \subset V^{(A)}\times V^{(A)}$ and
          $E^{(B)} \subset V^{(B)}\times V^{(B)}$ 
  must contain binary trees which compute each linear combination of inputs to the products computed by $\bilalg F$.
The parts $E^{(AC)} \subseteq V^{(A)}\times V^{(M)}$ and
          $E^{(BC)} \subseteq V^{(B)}\times V^{(M)}$ 
  define the inputs to each multiplication (so, if $\exists (v,w)\in E^{(AC)}, w\in V^{(M)}$ then $\exists (z,w)\in E^{(BC)}$ for some $z$).
Finally, the part $E^{(C)} \subset V^{(C)}\times (V^{(C)}\setminus V^{(M)})$ encodes the summations that compute the output elements (subset of $V^{(C)}$).
%For $X\in\{A,B,C\}$,
%$E^{(X)}\subset (V^{(X)}\cup W^{(X)})\times W^{(X)}$ defines the binary tree computing all linear combinations of elements from domain $D^{(X)}$ (leaves $V^{(X)}$).
%For $Y\in\{A,B\}$,
%$E^{(YM)}\subset (V^{(Y)}\cup W^{(Y)})\times W^{(M)}$ defines which linear combination of elements from input domain $D^{(Y)}$ serves as input to the multiplication computing each bilinear form (vertex in $W^{(M)}$).
%Finally, there must be a vertex set $V^{(C)}\subseteq W^{(M)}$, $|V^{(C)}|=|D^{(C)}|$, such that the binary tree of operations leading to this vertex from $V^{(A)}$ and $W^{(A)}$ computes the bilinear forms contributing to a unique element of $D^{(C)}$ and their correct linear combination as specified by $\bilalg F$.
\end{definition}

To quantify the communication requirements of an execution DAG, it helps to reason about its expansion properties.
Graph expansion is a known technique for derivation of communication lower bounds, which has previously been used for Strassen's algorithm~\cite{Ballard:SPAA2011,scott2015complexity,bilardi2016complexity}.
We employ a different but related notion of expansion, well-suited for bilinear algorithms.
We seek to bound the size of the boundary of a subset of vertices of the execution DAG.
% number of vertices connected by an edge in the DAG to a vertex within a given subset of vertices.
%To work with path connectivity we recall that any DAG may be associated with a partial ordering, and work with the relation `$\leq$', so given a graph $G=(V,E)$, we say $v\leq w$ for $v,w\in V$ if for some $k\geq 1, \exists \{v_1,\ldots,v_k\}$ such that $v_k=w$ and $(v_i,v_{i+1})\in E$ for $i\in [1,k-1]$.
\begin{definition}
For execution DAG $G$ defined and partitioned as in Definition~\ref{def:exeDAG}, we say that a subset of vertices that correspond to computed values, $Z\subset V\setminus V^{(I)}$ has
\begin{itemize}
\item $A$-expansion $\zeta^{(A)}_G(Z)$ if the number of external vertices in $V^{(A)}$ connected to $Z$ is  
\[\zeta^{(A)}_G(Z)\defeq |\{v : (v,w)\in E, v\in V^{(A)}\setminus Z, w \in Z \}|,\]
\item $B$-expansion $\zeta^{(B)}_G(Z)$ if the number of external vertices in $V^{(B)}$ connected to $Z$ is
\[\zeta^{(B)}_G(Z)\defeq |\{v : (v,w)\in E, v\in V^{(B)}\setminus Z, w \in Z \}|,\]
\item $C$-expansion $\zeta^{(C)}_G(Z)$ if the number of vertices in $Z$ that are in $V^{(C)}$ or connected to a vertex in $V^{(C)}\setminus Z$ (these subsets may overlap) is  
%connected to a vertex in $V^{(M)}\setminus Z$ or inside $V^{(C)}$ is $c^{(C)}$,  
\[\zeta^{(C)}_G(Z)\defeq |(Z\cap V^{(C)})\cup \{w : (w,v)\in E, v\in V^{(C)}\setminus Z, w \in Z\}|.\]
% \vee \exists w\in V^{(M)}\setminus Z, (v,w)\in E^{(C)} ]\}|.\]
\end{itemize}
\end{definition}
$A$-expansion and $B$-expansion count vertices that are \rev{part of} the outer boundary of $Z$, while $C$-expansion counts vertices that are part of the inner boundary of $Z$.
If an algorithm computes all the elements in $Z$, it must obtain inputs counted by $A$-expansion and $B$-expansion as well as yield outputs that are counted by $C$-expansion.
We will use these notions of expansion to bound the number of inputs and outputs, and consequently the communication costs, associated with computing any \rev{subgraph} of an execution DAG.

\subsection{Sequential Schedule Cost Model}
\label{sec:seq_model}

A sequential schedule imposes a total ordering on a (partially-ordered) execution DAG, interleaving the operations described by vertices in the execution DAG with reads from memory to cache, writes to cache from memory, and discards from cache.
To measure the vertical communication cost on a sequential computer, we consider a cache of size $\cs$ elements and assume all data starts in main memory.
We employ an idealized cache model, i.e., we do not consider track/cache-line size or mechanisms such as cache associativity.
We assign reads and writes of data between main memory and cache unit communication cost.
We do not pay attention to synchronization/latency cost.
%These assumptions on our analysis are justified as in the tensor contraction algorithms considered, all scalar multiplications can be done concurrently, and so they are generally latency insensitive.
%contain a very high amount of concurrency and have logarithmic or constant depth as matrix multiplication.
%\fixme{Edgar}{Cite...}

We refer to this sequential machine model as $\seqarch{\cs}$.
We do not restrict the size of the main memory of $\seqarch{\cs}$.
We do not allow schedules executing on $\seqarch{\cs}$ to recompute any element computed by the algorithm. % (although the algorithm can be defined to compute equivalent elements).
We assume none of the inputs of the algorithm reside in cache at the start of execution and that all of the outputs must be written to memory.
%\begin{itemize}
%\item read an element from memory to cache,
%\item write an element of data to memory,
%\end{itemize}
%We allow accumulation at the same cost as a read or write for simplicity, one could alternatively perform it via a read followed by a write.
We denote the sequential communication cost of a schedule executed on $\seqarch{\cs}$ as $\bwcost$ and provide lower bounds for the cost on $\bwcost$ for a given algorithm
by considering all valid sequential schedules of this algorithm.

\subsection{Parallel Schedule Cost Model}
\label{subsec:par_model}

%Our lower bounds assume that the input and outputs are load balanced.
%The memory size and horizontal (interprocessor) communication costs are also implicitly parameterized by tensor element size. % (all elements are assumed to be part of some Abelian group $R$).

\sloppy{We also consider \rev{parallel schedules} on a distributed-memory computer with a fully connected network.}
We denote this homogeneous parallel computer of $p$ processors as $\pararch{p}{M}$.
%, each capable of storing $M$ elements (i.e., with $M$ memory). 
%\fixme{Edgar}{Not currently doing anything with $M$ (keeping things messy until decision made as to whether we actually want to).}
We assume that each element of the input to the algorithm exists on a unique processor at the start of the execution of any parallel schedule and that the parallel schedule does not
compute any element twice (no recomputation).
We allow all processors to communicate with each other (fully connected network) on $\pararch{p}{M}$ via point-to-point messages. % and assume the communication is synchronous.
We measure the parallel horizontal communication cost $W$ of a schedule on $\pararch{p}{M}$ as the largest number of elements sent and received by any processor throughout the execution of the parallel schedule.
%This simple horizontal communication cost measurement does not consider dependent sequences of executions (i.e., idle time), however, the algorithms we consider have constant depth
%(short critical path length) and by deriving lower bounds on this basic cost measurement we also obtain lower bounds on the cost of the algorithm in other models.
%In particular, by obtaining lower bounds on the amount of communication done by any processor for any schedule of an algorithm, we obtain lower bounds on horizontal communication bandwidth cost 
Lower bounds on this simple communication metric yield lower bounds 
for LogP~\cite{Culler:1993}, LogGP~\cite{Alexandrov:1995}, BSP~\cite{valiant1990bridging}, and the $\alpha$--$\beta$ critical path cost 
model (described in detail in~\cite{SCKD_TECHREP_2013}).
In all of these models, the communication cost of a parallel schedule is at least the communication cost incurred by any given processor.
%We denote the parallel communication cost of an algorithm as $W$ (usually with a subscript corresponding to the name of the algorithm).

We assume throughout our analysis that the number of processors $p$ divides into the input/output domain sizes and the total number of operations.
In our horizontal communication lower bounds, we treat $n$ (tensor dimension) and $p$ as asymptotic parameters, while assuming the tensor orders given by $s,t,v$ are constants.
\begin{definition}
We say a parallel schedule is {\bf storage-balanced} if
\begin{itemize} 
\item at the start of execution, each processor owns $x/p$ of the elements of each operand of size $x$;
\item at the end of execution, each processor owns $y/p$ of the elements of the output of size $y$.
\end{itemize}
\end{definition}
Storage-balanced schedules can have any initial distribution of inputs and final distribution of outputs, so long as it is minimal in memory usage at the start and end of execution, and each element is stored on a unique processor.

\section{Lower Bounds for Bilinear Tensor Algorithms}
\label{sec:bi_lb}

In order to derive non-trivial lower bounds on the communication costs of a bilinear algorithm, we need to know something about the sparsity structure of the tensors in $\mathcal{F}$.
In particular, we desire a lower bound on the number of inputs and outputs required by any subset of the products that the bilinear algorithm computes.
Further, we want to know the minimal number of linear combinations that can represent the inputs and the outputs, motivating the use of matrix rank for our analysis.
\begin{definition}\label{def:biexp}
A bilinear algorithm $\bilalg F$ has nondecreasing (in all variables) \rev{an expansion bound function} $\mathcal{E}_\Lambda : \mathbb{N}^3\to \mathbb{N}$, if for all $\bilalg R \subseteq \bilalg F$,
%with matrix form $\mnnbilalg stvHR\lambda$,
\[\rank(\bilalg R ) \leq \mathcal{E}_\Lambda\lt(\rank(\B{R}^{(A)}),\rank(\B{R}^{(B)}),\rank(\B{R}^{(C)})\rt).\]
\end{definition}

%All  expansion bounds will be nondecreasing in all three variables.
The  expansion bound can be used to obtain lower bounds on the amount of communication needed to obtain the necessary inputs and produce the necessary outputs of an arbitrary set of computed products.
The ranks of the matrices $\B{R}^{(A)}$, $\B{R}^{(B)}$, and $\B{R}^{(C)}$ may be used to obtain lower bounds on the amount of input/output tensor data that is associated with computing the products.

\begin{lemma}\label{lem:exp_bi}
For execution DAG $G=(V,E)$ (partitioned as in Definition~\ref{def:exeDAG}) of bilinear algorithm $\bilalg F$, consider any computed subset $Z\subseteq V\setminus V^{(I)}$. 
If $\bilalg F$ has expansion bound $\mathcal{E}_\Lambda$, %it must hold that,
\[|Z\cap V^{(M)}| \leq \mathcal{E}_\Lambda(\zeta^{(A)}_G(Z),\zeta^{(B)}_G(Z),\zeta^{(C)}_G(Z)).\]
%
%Any schedule snapshot that computes a subset of bilinear forms $V\subseteq D^{(M)}$ of a bilinear algorithm $\bilalg F$ with an  expansion bound $\mathcal{E}_\Lambda$, that takes a vector of inputs $\B{v}^{(A)}$ with capacitance $c^{(A)}=\zeta^{(A)}(\B{v}^{(A)})$ and a vector of inputs $\B{v}^{(B)}$ with capacitance $c^{(B)}=\zeta^{(B)}(\B{v}^{(B)})$, and produces a vector of outputs $\B{v}^{(C)}$ with capacitance $c^{(C)}=\zeta^{(C)}(\B{v}^{(C)})$, 
%%`$a$' linear combinations of elements of $\C A$, `$b$' linear combinations of elements of $\C B$ are required, and `$c$' distinct contributions to elements of $\C C$ must be produced, for some $a,b,c\in\mathbb{N}$, 
%must satisfy $|V| \leq\mathcal{E}_\Lambda\lt(c^{(A)},c^{(B)},c^{(C)}\rt)$. %\zeta^{(A)},\zeta^{(B)},\zeta^{(C)}\rt)$.
\end{lemma}
\begin{proof}
Consider the bilinear algorithm restriction $\bilalg R\subseteq \bilalg F$, \rev{acting on inputs $\B a$ and $\B b$,} where the matrices in $\mathcal{R}$ contain the subset of columns of the matrices in $\mathcal{F}$ corresponding to the products $H^{(M)}=Z\cap V^{(M)}$.
By Definition~\ref{def:biexp}, we must have 
%we know that the matrices in its matrix form ($\mnnbilalg stvHR\lambda$), $\C{\hat{R}^{(A)}}$, $\C{\hat{R}^{(B)}}$, and $\C{\hat{R}^{(C)}}$ must have ranks so that,
\[|H^{(M)}| \leq \mathcal{E}_\Lambda\lt(\rank(\B{R}^{(A)}),\rank(\B{R}^{(B)}),\rank(\B{R}^{(C)})\rt).\]
Now, $\zeta^{(A)}_G(Z)$ counts the number of linear combinations of elements in $\B a$ on the boundary of $Z$ within $G$.
%Since $Z$ may not contain input vertices, 
All \rev{linear combinations of elements of $\B a$ needed to compute the}
products in $H^{(M)}$ must be computed using these $c^{(A)}$ linear combinations of $\B a$ elements, $\B{v}^{(A)}$.
For any input vector $\B{{a}}$ to $\bilalg R$, 
\rev{these linear combinations are described by columns of $\B{R}^{(A)}$, i.e., we compute}
$\B{w}^{(A)}=\B{R}^{(A)}{}^T\B{{a}}$.
Further, since $\B{v}^{(A)}$ are linear combinations, there exists a matrix $\B{L}$ such that, $\B{v}^{(A)}=\B{L}\B{{a}}$ and since all products in $H^{(M)}$ are computed from $\B{v}^{(A)}$, there exists a matrix $\B{M}$ such that $\B{w}^{(A)}=\B{M}\B{v}^{(A)}=\B{M}\B{L}\B{{a}}$.
Therefore, we must have $\rank(\B{R}^{(A)})= \rank(\B{M}\B{L})$, and since $\rank(\B{M}\B{L})\leq \rank(\B{M})\leq \zeta^{(A)}_G(Z)$, we have that $\rank(\B{R}^{(A)})\leq \zeta^{(A)}_G(Z)$.

An identical argument may be made to show that $\rank(\B{R}^{(B)})\leq\zeta^{(B)}_G(Z)$.
The quantity $\zeta^{(C)}_G(Z)$ is the number of output elements computed within $Z$ and the number of bilinear forms within $Z$ that are dependencies of output elements or other bilinear forms not included in $Z$.
Since the execution DAG may not recompute bilinear forms, by the same argument as for $\zeta^{(A)}_G(Z)$, we assert that $\rank(\B{R}^{(C)})\leq \zeta^{(C)}_G(Z)$.
Otherwise the execution DAG may not be valid (the schedule cannot compute the result of the given transformation due to its rank).
These rank bounds suffice to prove the theorem, as the expansion bound must be nondecreasing.
\end{proof}

%We now derive vertical and then horizontal communication lower bounds for arbitrary bilinear algorithms given an expansion function.
%To prove lower bounds for specific bilinear algorithms, we will derive expansion bounds via volumetric inequalities (Section~\ref{sec:vol_ineq}) and use them to apply the general bounds from this section.

\subsection{Vertical Communication Lower Bounds for Bilinear Algorithms}
\label{sec:seq_bi_lb}

Knowing the  expansion bound of a bilinear algorithm allows us to obtain a vertical communication lower bound for it by a straight-forward counting argument.
\begin{theorem}\label{thm:seq_comm_lwb_bi}
Any schedule on $\seqarch{\cs}$ of any execution DAG, $G=(V,E)$, of \rev{irreducible} bilinear algorithm $\bilalg F$ with $\rank(\bilalg F)=r^{(\Lambda)}$, $\dim(\bilalg F)=(r^{(A)},r^{(B)},r^{(C)})$, and expansion bound $\mathcal{E}_\Lambda$ has vertical communication cost,
{\mathsmall
\[
\bwcost_\Lambda(\mathcal{E}_\Lambda,r^{(\Lambda)},r^{(A)},r^{(B)},r^{(C)},\cachesize)\geq \max\lt[\frac{2r^{(\Lambda)} \cs}{\mathcal{E}_\Lambda^\mathrm{max}(H)},  r^{(A)}+ r^{(B)}+ r^{(C)} \rt],\]} \\
with
{\mathsmall \(\displaystyle{\mathcal{E}_\Lambda^\mathrm{max}(\cs) = \max_{c^{(A)},c^{(B)},c^{(C)}\in \mathbb{N}, c^{(A)}+c^{(B)}+c^{(C)}= 3\cs}\mathcal{E}_\Lambda(c^{(A)},c^{(B)},c^{(C)}),}\)}
under the assumption that $\mathcal{E}_\Lambda^\mathrm{max}(\cs)$ is strictly increasing and convex for $\cs\geq 1$.
%\fixme{Edgar}{Slight reinforcement possible, not clear if useful yet: $c^{(X)}\leq |D^{(X)}|$}
\end{theorem}
\begin{proof}
The second term in the maximum within the lower bound, $ r^{(A)}+ r^{(B)}+ r^{(C)}$, arises since any schedule must read all inputs into cache to operate on them and must write all outputs back to memory.
Since we assumed that the bilinear algorithm is irreducible, $ r^{(A)}+ r^{(B)}$ such inputs need to be read from memory and $r^{(C)}$ outputs must be written back. % (their computation is always explicitly specified).

The first term in the maximum within the lower bound $\frac{2r^{(\Lambda)}\cs}{\mathcal{E}_\Lambda^\mathrm{max}(H)}$, arises as a result of the expansion bound of the bilinear algorithm.
Any storage-balanced schedule on $\seqarch{\cs}$ that computes $\bilalg F$ may be partitioned into $f=\lfloor r^{(\Lambda)}/\mathcal{E}_\Lambda^\mathrm{\max}(H)\rfloor$ intervals such that in each interval exactly $\mathcal{E}_\Lambda^\mathrm{\max}(H)$ products are computed and a last interval in which the remaining $\bar{m}=r^{(\Lambda)} \bmod \mathcal{E}_\Lambda^\mathrm{\max}(H)$ forms are computed.
In each of the first $f$ intervals, if the set of computed vertices corresponding to products in the execution DAG is $Z\subset V^{(M)}$, 
by Lemma~\ref{lem:exp_bi}, 
\[\mathcal{E}_\Lambda(\zeta_G^{(A)}(Z),\zeta_G^{(B)}(Z),\zeta_G^{(C)}(Z))\geq |Z| = \mathcal{E}_\Lambda^\mathrm{max}(\cs).\]
Furthermore, each of these schedule intervals requires $\zeta_G^{(A)}(Z)$ linear combinations of $\B a$ and $\zeta_G^{(B)}(Z)$ linear combinations of $\B b$ to be input, and $\zeta_G^{(C)}(Z)$ linear combinations of products to be output.
By the definition of $\mathcal{E}_\Lambda^\mathrm{max}(\cs)$, we have $\zeta_G^{(A)}(Z)+\zeta_G^{(B)}(Z)+\zeta_G^{(C)}(Z)\geq 3\cs$, so the total number of inputs and outputs for each of these intervals must be at least $3\cs$. 
%computing the bilinear forms for each of the first $f$ intervals will require at least $3\cs$ inputs and outputs, since to get $\mathcal{E}_\Lambda(c^{(A)},c^{(B)},c^{(C)})\geq \mathcal{E}_\Lambda^\mathrm{max}$, we need $c^{(A)}+c^{(B)}+c^{(C)}\geq3\cs$ as $\mathcal{E}_\Lambda^\mathrm{max}$ is strictly increasing.
%Here $c^{(A)}+c^{(B)}$ is the number of inputs (linear combinations of elements either of $\C A$ or of $\C B$) and $c^{(C)}$ is the number of outputs (distinct contributions to elements of $\C C$, linear combinations of bilinear forms). 
Similarly, the last interval requires at least $3\bar{H}$ inputs and outputs, where $\bar{H}$ is the maximum integer such that $\mathcal{E}_\Lambda^{\mathrm{max}}(\bar{H})\leq \bar{m}$.
We additionally assert that
\[\bar{H} \geq \cs(\bar{m}/\mathcal{E}_\Lambda^{\mathrm{max}}(\cs)),\]
because $\bar{m}\geq \mathcal{E}_\Lambda^{\mathrm{max}}(\bar{H})$ and $\mathcal{E}_\Lambda^{\mathrm{max}}(\cs)/\mathcal{E}_\Lambda^{\mathrm{max}}(\bar{H})  \geq \cs/\bar{H}$, the latter due to
$\mathcal{E}_\Lambda^{\mathrm{max}}$ being strictly increasing and convex. 
%\fixme{Edgar}{(Partial sums?)}
%Without loss of generality, we assume $r^{(\Lambda)} \neq 0 \mod \mathcal{E}_\Lambda^\mathrm{\max}$, so that the last interval computes only 

Now, let $x_i$ be the number of inputs present in cache prior to execution of interval $i$
(by assumption no inputs start in cache so $x_1=0$). The number of inputs %(elements of $\C A$ and $\C B$)
present in cache at the end of execution of interval $i$ should be equal to the number of 
inputs available for interval $i+1$, $x_{i+1}$. Let the number of contributions to $\C C$ (bilinear forms),
which remain in cache (are not written to memory) at the end of interval $i$, be $y_i$.
The rest of the outputs produced by interval $i$ must be written to memory. 
Let $x_{f+1}$ be the number of inputs in cache prior to execution of the last interval,
by assumption all outputs must be written to memory by the end of this interval, so $y_{f+1}=0$. 
In total, the amount of reads from memory and writes to memory done during the execution of interval $i$ with $\mathcal{E}_\Lambda^\mathrm{\max}$ 
products is then at least $w_i\geq 3\cs-x_i-y_i, \forall i\in \inti{1}{f}$ and $w_{f+1}\geq 3\bar{H}-x_{f+1}$. 
Now, since the inputs and outputs which are kept in cache at the end of interval $i$ must fit in cache, we know that  
$\forall i\in\inti{1}{f}, \ x_{i+1}+y_i\leq \cachesize$.
Rearranging this and substituting $y_i$ into our bound on reads and writes, we obtain $w_i\geq 2\cs-x_i+x_{i+1}, \forall i\in \inti{1}{f}$. 
Summing over all intervals and extracting the first interval to apply $x_1=0$ yields the desired lower bound on total communication cost, 
\begin{align*}
 &\bwcost_\Lambda(\mathcal{E}_\Lambda,r^{(\Lambda)},r^{(A)},r^{(B)},r^{(C)},\cachesize)\ \geq\ \sum_{i=1}^{f+1} w_i  \\
&\ \geq\  2\cs + x_1 + \sum_{i=2}^{f}(2\cs-x_i+x_{i+1}) + 3\bar{H}-x_{f+1} 
\ =\  2f\cs +3\bar{H} \\
 &\ \geq\ 2\cs\lt\lfloor r^{(\Lambda)}/\mathcal{E}_\Lambda^\mathrm{\max}(H)\rt\rfloor + 3\cs (\bar{m}/\mathcal{E}_\Lambda^{\mathrm{max}}(\cs)) \ \geq\  \frac{2r^{(\Lambda)}\cs}{\mathcal{E}_\Lambda^\mathrm{\max}(H)}. \qed
\end{align*}
%\vspace*{-.25in}
\renewcommand{\qedsymbol}{}
\end{proof}
This vertical communication lower bound implies that, aside from moving the inputs and outputs between memory and cache, any execution of the bilinear algorithm requires additional communication that is dependent on its bilinear expansion.
In particular, a bilinear expansion lower bound can be maximized to upper-bound the number of products that can be computed given any set of elements that fits in cache.
Such an upper bound consequently yields the lower bound on communication with respect to the total amount of products computed.

\subsection{Horizontal Communication Lower Bounds for Bilinear Algorithms}
\label{sec:par_bi_lb}

We can also formulate a lower bound on horizontal communication cost based on expansion bounds and the dimensions of the inputs.
\begin{theorem}\label{thm:par_comm_lwb_bi}
Any storage-balanced schedule on $\pararch{p}{M}$ of any execution DAG of \rev{irreducible} bilinear algorithm $\bilalg F$ with $\rank(\bilalg F)=r^{(\Lambda)}$, $\dim(\bilalg F)=(r^{(A)},r^{(B)},r^{(C)})$, and expansion bound $\mathcal{E}_\Lambda$ has horizontal communication cost,
\[W_\Lambda(\mathcal{E}_\Lambda,r^{(\Lambda)},r^{(A)},r^{(B)},r^{(C)},p)\geq c^{(A)}+c^{(B)}+c^{(C)}\]
for some non-negative $c^{(A)},c^{(B)},c^{(C)}\in \mathbb{N}$ such that 
\[\displaystyle{r^{(\Lambda)}/p \leq \mathcal{E}_\Lambda(c^{(A)}+ r^{(A)}/p, c^{(B)}+ r^{(B)}/p, c^{(C)}+ r^{(C)}/p)}.\]
\end{theorem}
\begin{proof}
For any execution DAG $G=(V,E)$ of $\bilalg F$, at least one processor must compute \rev{at least} $r^{(\Lambda)}/p$ products.
Let $Z\subset V\setminus V^{(I)}$ be the set of all linear combinations of elements of $\B a$ and $\B b$ and bilinear forms this processor computes, so $|Z\cap V^{(M)}|\geq r^{(\Lambda)}/p$.
The $A$- and $B$- expansions of $Z$ count all linear combinations of input elements that are operands of products or additions computed by the processor, but are not themselves computed by this processor.
Such linear combinations must either be input initially or received via a message.
The $C$-expansion counts all bilinear forms computed by this processor that are outputs or operands to sums computed by other processors.
Such bilinear forms must be output or sent in an outgoing message by this processor.
%The $A$-expansion of $Z$ counts all elements that are inputs to operations computed by the processor, but are not computed by this processor, meaning they must either be initially input or received via a message.
\rev{By Lemma~\ref{lem:exp_bi}, we can relate the number of products computed by this processor to these expansion set sizes by $r^{(\Lambda)}/p\leq\mathcal{E}_\Lambda(\zeta^{(A)}_G(Z),\zeta^{(B)}_G(Z),\zeta^{(C)}_G(Z))$.}
% number of vertices that are inputs to the subgraph of $G$ induced by the vertex subset $Z$.
%In particular, if for the processor computing vertices $Z$,
%\begin{itemize}
%\item the $A$-expansion (inputs or received elements) is $d^{(A)}$,
%\item the $B$-expansion (inputs or received elements) is $d^{(B)}$,
%\item the $C$-expansion (outputs or sent elements) is $d^{(C)}$,
%\end{itemize}
%Some processor must compute a set of products (meaning the processor performs the needed multiplications) of size at least $r^{(\Lambda)}/p$.
%For this processor, let the number of inputs (initially-owned or received elements) of $\B A$ be $d^{(A)}$ and of $\B B$ be $d^{(B)}$ and the number of outputs (owned at finish or sent elements) of the processor be $d^{(C)}$.
%%We can then consider the schedule snapshot corresponding to everything this processor computes,
%%where the inputs are $\B{v}^{(A)},\B{v}^{(B)}$ with capacitances 
%%$c^{(A)}=\zeta^{(A)}(\B{v}^{(A)})$,
%%%\geq  r^{(A)}/p$ 
%%and $c^{(B)}=\zeta^{(B)}(\B{v}^{(B)})$,
%%%\geq  r^{(B)}/p$,
%%while the outputs $\B{v}^{(C)}$ has capacitance 
%%$c^{(C)}=\zeta^{(C)}(\B{v}^{(C)})$.
%%%\geq  r^{(M)}/p$.
%By Lemma~\ref{lem:exp_bi}, we can lower bound the size of the computed set of products with respect to the number of these inputs and outputs, 
%then $r^{(\Lambda)}/p\leq\mathcal{E}_\Lambda(d^{(A)},d^{(B)},d^{(C)})$.

Now, since we assume schedules on $\pararch{p}{M}$ must have load-balanced inputs and outputs, 
this processor owns $ r^{(A)}/p$ elements of $\B a$ and $ r^{(B)}/p$ elements of $\B b$ at the
start of execution, and outputs $ r^{(C)}/p$ elements of $\B c$.
Thus, the processor computing $Z$ must receive at least $c^{(A)}=\zeta^{(A)}_G(Z)- r^{(A)}/p$ linear combinations of elements of $\B{a}$, and at least $c^{(B)}=\zeta^{(B)}_G(Z)- r^{(B)}/p$ linear combinations of elements of $\B{b}$.
Further, it outputs $r^{(C)}/p$ elements of $\B c$, so it must send at least $c^{(C)}=\zeta^{(C)}_G(Z)- r^{(C)}/p$ bilinear forms.
%(otherwise some of the products computed by this processor would have had to be recomputed, which is not allowed by our execution model).
Thus, for some $c^{(A)},c^{(B)},c^{(C)}\in\mathbb{N}$, we have
\mathsmall{
\[\frac{r^{(\Lambda)}}{p} \leq \mathcal{E}_\Lambda(\zeta^{(A)}_G(Z),\zeta^{(B)}_G(Z),\zeta^{(C)}_G(Z)) = \mathcal{E}_\Lambda\Big(c^{(A)}+ \frac{r^{(A)}}p, c^{(B)}+ \frac{r^{(B)}}p, c^{(C)}+ \frac{r^{(C)}}p\Big),\]}\\
with a corresponding communication cost lower bound of
\begin{align*}
W_\Lambda(\mathcal{E}_\Lambda,r^{(\Lambda)},r^{(A)},r^{(B)},r^{(C)},p)\geq c^{(A)}+c^{(B)}+c^{(C)}. \qed
\end{align*}
%\vspace*{-.25in}
\renewcommand{\qedsymbol}{}
\end{proof}
\rev{This horizontal communication lower bound expresses the minimum amount of data that needs to be communicated ($c^{(A)}+c^{(B)}+c^{(C)}$), given that each processor starts with a $r^{(A)}/p+r^{(B)}/p$ inputs and ends up with $r^{(C)}/p$ outputs.}
The amount of communication depends on the rank as well as the expansion of the bilinear algorithm.

\section{Lower Bounds for Symmetric Tensor Contraction Algorithms}
\label{sec:sypr_lb}

The bilinear algorithm lower bound infrastructure introduced above allows us use expansion bounds to derive lower bounds that reproduce known results for nonsymmetric tensor contractions, as well as derive new asymptotically tighter bounds for symmetric tensor contraction algorithms.
To establish bounds on matrix rank in these bilinear algorithms, we make use of volumetric inequalities.
In particular, we use the Loomis-Whitney inequality, which is standard for communication lower bounds in matrix computations, as well as a higher-order generalization thereof (see Appendix~\ref{sec:vol_ineq}).
%The vertical and horizontal communication lower bounds for bilinear algorithms may be applied to 
For matrix multiplication, nonsymmetric tensor contraction algorithms, and direct evaluation of symmetric contractions, we provide communication lower bounds that match known results (see Appendices~\ref{sec:ns_lb} and~\ref{sec:syde_lb}).
In this section, we focus on derivation of communication lower bounds that imply asymptotically more communication due to use of symmetry.
First, we derive a horizontal communication lower bound for the direct evaluation algorithm (which can be asymptotically higher than the bound derived in Appendix~\ref{sec:syde_lb}).
Then, we derive vertical and horizontal communication lower bounds for the symmetry preserving algorithm, which can also exceed their nonsymmetric counterparts.

%Consequently, in this section, we focus on bounds that are asymptotically higher than their nonsymmetric counterparts, providing other results in appendices.
%We apply a similar lower bound approach for
%For the direct evaluation algorithm of symmetric tensor contractions, we obtain communication lower bounds that are somewhat smaller than the nonsymmetric case.
%This result (derived in Appendix~\ref{sec:syde_lb}) is expected, since the direct evaluation algorithm requires less computation than the nonsymmetric algorithm for most contractions.
%Further, additional symmetric equivalence between elements exists within these matrices, which makes a lower communication cost possible.

\subsection{Horizontal Communication Lower Bound for the Direct Evaluation Algorithm}
\label{subsec:sym_de_lb_main}

For matrix-vector-like contractions, we provide an expansion bound for the direct evaluation algorithm that leverages our assumption that each input and output entry must be stored on a single processor.
%does not have a counterpart for the nonsymmetric algorithm.
%We now also derive an additional expansion bound for the direct evaluation algorithm for symmetric contractions.
%This bound is valid when exactly one of $s,t,v$ is zero (the contraction is matrix-vector-like), and in some cases can be stronger than the matrix-multiplication-like bound above.
We use this bound to demonstrate the communication overhead of symmetrization and symmetric storage, where symmetrically-equivalent entries are not stored redundantly.
% it to prove a horizontal communication lower bound in Theorem~\ref{thm:comm_lowerb_syst_prl2}.
\begin{lemma}\label{lem:strd_exp2}
An expansion bound %(Definition~\ref{def:biexp}) 
on $\nsystalg stv$ when exactly one of $s,t,v$ is zero is
\mathfootnotesize{
\begin{align*}
&\mathcal{E}^{(s,t,v)}_\Psi(d^{(A)},d^{(B)},d^{(C)})=
\bigg({\omega \choose \min(s,v)}-1\bigg)d^{(A)}+
\bigg({\omega \choose \min(v,t)}-1\bigg)d^{(B)} \\&+
\bigg({\omega \choose \min(s,t)}-1\bigg)d^{(C)} 
+\min\bigg((d^{(A)})^{\omega/(s+v)}, (d^{(B)})^{\omega/(v+t)}, (d^{(C)})^{\omega/(s+t)}\bigg).
\end{align*}
}
\end{lemma}
\begin{proof}
Let $\nbilalg stvDK = \nsystalg stv$.
We prove the case when $v=0$ (the other two are similar).
Consider any $\bilalg R\subseteq \nbilalg stvDK$ and the associated subset of products, $V\subseteq D^{(M)}=\enti{1}{n}s\enti{1}{n}t$.
Let $c_1$ be the number of rows of $\B{{R}^{(C)}}$ that contain ${\omega}\choose s$ nonzeros \rev{(the maximal number possible)}.
Let $c_2$ be the number of other nonzero rows of $\B{{R}^{(C)}}$ (that contain at least one and at most ${\omega\choose s} -1$ nonzeros).
Since, when $v=0$, each column in $\B{R^{(C)}}$ has a single nonzero, $d^{(C)}\defeq \rank(\B{R^{(C)}})=c_1+c_2$, and further,
\[|V|\leq \bigg({\omega \choose s}-1\bigg)c_1 + {\omega \choose s}c_2 = \bigg({\omega \choose s}-1\bigg)d^{(C)} + c_1.\]
%When $v=0$ all products contribute to a single entry of $\C{C}$, therefore,
%\[|V|\leq \bigg({\omega \choose s}-1\bigg)d^{(C)}+c.\]
Let $Z\subseteq\enti 1n{\omega}$ be the set of indices of the $c_1$ rows ($\B C$ entries) with ${\omega\choose s}$ nonzeros.
Since the bilinear algorithm \rev{restriction} computes all products contributing to these $c_1$ rows, its $\bhRA$ and $\bhRB$ matrices have nonzeros in all rows which correspond to entries of $\C A$ and $\C B$ contributing to these $c_1$ entries of $\C{C}$.
Therefore, the set of rows with nonzeros in $\bhRA$ is
\[L^{(A)}= \{\tpl js : \tpl js \in \prj s(\tpl i{\omega}), \tpl i{\omega} \in Z \}.\]
Since each column of $\bhRA$ has a single nonzero, any set of unique rows is linearly independent, which implies that the rank of $\bhRA$ is at least $|L^{(A)}|$.
%as each column of $\bhRA$ has a single nonzero (so sets of unique columns are linearly independent).
By Lemma~\ref{thm:tlw}, $|L^{(A)}|\geq |Z|^{s/\omega}$, so $d^{(A)}\geq |Z|^{s/\omega}=c_1^{s/\omega}$.
By a similar argument, the rank of $\bhRB$ is at least $d^{(B)}\geq c_1^{t/\omega}$.

Therefore, since $c_1\leq\min((d^{(A)})^{\omega/s}, (d^{(B)})^{\omega/t}, d^{(C)})$, given any subset bilinear algorithm where the ranks of $\bhRA$, $\bhRB$, and $\bhRC$ being $d^{(A)}$, $d^{(B)}$, and $d^{(C)}$, respectively, we can bound the rank of the subset bilinear algorithm by the formula,
{\mathfootnotesize
\begin{align*}
\mathcal{E}^{(s,t,0)}_\Psi(d^{(A)},d^{(B)},d^{(C)})=\bigg({\omega \choose s}-1\bigg)d^{(C)}+\min\bigg((d^{(A)})^{\omega/s}, (d^{(B)})^{\omega/t}, d^{(C)}\bigg).
\end{align*}
} \\
%The above formula bounds the maximum number of products evaluated (subset algorithm rank) with respect to the maximum $c\in\inti{1}{\min((d^{(A)})^{\omega/s},(d^{(B)})^{\omega/t},d^{(C)})}$.
Further, there are symmetric bounds for other permutations of $s,t,v$, 
%so $\mathcal{E}^{(s,0,v)}_\Psi (d^{(A)},d^{(B)},d^{(C)})=\mathcal{E}^{(s,v,0)}_\Psi (d^{(A)},d^{(C)},d^{(B)})$ and
%$\mathcal{E}^{(0,t,v)}_\Psi (d^{(A)},d^{(B)},d^{(C)})=\mathcal{E}^{(v,t,0)}_\Psi (d^{(C)},d^{(B)},d^{(A)})$.
{\mathfootnotesize
\begin{align*}
\mathcal{E}^{(0,t,v)}_\Psi(d^{(A)},d^{(B)},d^{(C)})&=\bigg({\omega \choose t}-1\bigg)d^{(B)}+\min\bigg((d^{(A)})^{\omega/v}, d^{(B)},  (d^{(C)})^{\omega/t}\bigg), \\
\mathcal{E}^{(s,0,v)}_\Psi(d^{(A)},d^{(B)},d^{(C)})&=\bigg({\omega \choose s}-1\bigg)d^{(A)}+\min\bigg(d^{(A)}, (d^{(B)})^{\omega/v}, (d^{(C)})^{\omega/s}\bigg).
\end{align*}
} \\
We can generalize the second term in all of these bounds by 
\[\min\bigg((d^{(A)})^{\omega/(s+v)}, (d^{(B)})^{\omega/(v+t)}, (d^{(C)})^{\omega/(s+t)}\bigg),\]
when exactly one of $s,t,v$ is zero.
We similarly generalize the first term in each of the bounds,
% so that only one of the three is nonzero whenever exactly one of $s,t,v$ is zero, 
yielding the bound in the lemma.
%rewriting the first terms so that they are nonzero, only when
%Combining these three bound permutations yields the general bound given in the theorem.
%\begin{align*}
%\mathcal{E}^{(s,t,0)}_\Psi(d^{(A)},d^{(B)},d^{(C)})=\max_{c\in\inti{1}{d^{(C)}}}\bigg[\bigg(&{s+t \choose s}-1\bigg)(d^{(C)}-c) \\
%&+{s+t\choose s}\min\bigg(c,(d^{(A)})^{(s+t)/s}, (d^{(B)})^{(s+t)/t}\bigg)\bigg].
%\end{align*}
%in $c{\omega\choose s}$ different columns, since each bilinear form (column) contributes to a unique entry of $\C C$.
\end{proof}
This expansion bound yields a new horizontal communication lower bound for matrix-vector-like contractions.
The additional bound is stronger than the bound obtained via the reduction from matrix multiplication (Theorem~\ref{thm:comm_lowerb_syst_prl1}) when $s,t,v$ are all unequal (and one is zero).
In the sequential scenario, the communication cost of such cases is dominated by reading the inputs from memory or writing the output to memory.
In the parallel scenario, the largest tensor can be kept local to each processor, so horizontal communication is generally dominated by moving the entries of the second-largest tensor.

When $v=0$, the new communication lower bound then arises as a consequence of symmetrization needed to compute $\B C$. 
When $s=0$ or $t=0$, this bound is a consequence of the assumption that each unique elements of the largest symmetric operand is stored on a unique processor.
%are not stored redundantly (i.e., each unique entry exists on a unique processor).
%to our assumption that only unique entries of the tensors are stored at the beginning of execution and each entry
%of the output is only computed once.
This assumption corresponds to a `packed' distributed layout~\cite{SMHSD_JPDC_2014} for symmetric tensors. Therefore, using an `unpacked' layout, where the tensor operands are stored as if they were nonsymmetric,
% by storing each symmetrically equivalent entry in the tensor on different processors. 
%  that each unique entry is stored on a unique processor (e.g., storing all entries of the symmetric operands) 
would make it possible to have asymptotically lower communication costs in cases when either $s$ or $t$ is zero and $v\neq s$ as well as $v\neq t$.
%From a practical perspective, this implies that packed distributed storage of tensors (which minimizes memory)~\cite{SMHSD_JPDC_2014} 
%has its disadvantages for certain contractions.
\begin{theorem}\label{thm:comm_lowerb_syst_prl2}
When exactly one of $s,t,v$ is zero, any storage-balanced schedule of $\systalg{\C A}{\C B}stv$ on $\pararch{p}{M}$ has horizontal communication cost, %a parallel machine with $p$ processors each with $M$ memory
%i.e., the number of words moved between the processors,
%under the assumptions that the work is balanced is %\fixme{Edgar}{make this precise} is
%(every processor does $O(F^\Upsilon(\varrho,n,s,t,v)/p)$ work)
%and that all elements have unit element size has the following lower bound,
%$W_\Psi(n{,}s{,}t{,}v{,}p{,}M)$
%\begin{align*}
%=\Omega\lt(W_\mathrm{MM}\lt(\chchoose{n}{s},\chchoose{n}{t},\chchoose{n}{v},p,qM\rt)\rt),
%\end{align*}
%where $q=\max\lt({s+v\choose s},{v+t\choose t},{s+t \choose s}\rt)$.
%Further, 
%such cases correspond to the third scenario is Theorem~\ref{thm:par_comm_lwb_mm}, where the 
%smallest matrix (with dimensions $x$ and $y$) is communicated.
\[W_\Psi(n,s,t,v,p) = \Omega\lt((n^\omega/p)^{\max(s,t,v)/\omega}\rt).\]
%\fixed{Edgar}{The above is only true for $m,n,k\gg 1$, do the other cases.}
%\fixed{Edgar}{Prove this with constants using Loomis-Whitney. Nah}
\end{theorem}
\begin{proof}
When $v=0$, by Lemma~\ref{lem:strd_exp2}, we have the the bound
{\mathfootnotesize $\mathcal{E}^{(s,t,0)}_\Psi(d^{(A)},d^{(B)},d^{(C)})=\big({\omega \choose s}-1\big)d^{(C)}+\min\big[(d^{(A)})^{\omega/s}, (d^{(B)})^{\omega/t}, d^{(C)}\big]$.}
Applying Theorem~\ref{thm:par_comm_lwb_bi}, we obtain 
%with this expansion bound, we obtain the bound,
\[W_\Psi(n,s,t,0,p)\geq d^{(A)}+d^{(B)}+d^{(C)},\]
for some $d^{(A)},d^{(B)},d^{(C)}\in \mathbb{N}$ such that 
\begin{align*}
&\frac{\chchoose{n}{s}\chchoose{n}{t}}{p} \leq 
\bigg({\omega \choose s}-1\bigg)\bigg(d^{(C)}+\chchoose{n}{\omega}/p\bigg) \\
&+\min\bigg[\bigg(d^{(A)}+\chchoose{n}{s}/p\bigg)^{\omega/s}, \bigg(d^{(B)}+\chchoose{n}{t}/p\bigg)^{\omega/t}, d^{(C)}+\chchoose{n}{\omega}/p\bigg].
\end{align*}
Using the fact that ${\omega \choose s}\chchoose{n}{\omega}\leq \chchoose ns\chchoose nt$, we can simplify the above constraint to
{\mathfootnotesize
\begin{align*}
\frac{\chchoose{n}{\omega}}p \leq&
\bigg({\omega \choose s}{-}1\bigg)d^{(C)} 
{+}\min\bigg[\bigg(d^{(A)}{+}\frac{\chchoose{n}{s}}p\bigg)^{\omega/s}, \bigg(d^{(B)}{+}\frac{\chchoose{n}{t}}p\bigg)^{\omega/t}, d^{(C)}{+}\frac{\chchoose{n}{\omega}}p\bigg].
%\bigg({\omega \choose s}-1\bigg)d^{(C)} \\
%&+\min\bigg[\bigg(d^{(A)}+\chchoose{n}{s}/p\bigg)^{\omega/s}, \bigg(d^{(B)}+\chchoose{n}{t}/p\bigg)^{\omega/t}, d^{(C)}+\chchoose{n}{\omega}/p\bigg].
\end{align*}
}
If the first term in the right hand side equation is greater than the min, we have $d^{(C)} = \Omega(n^{\omega}/p)$, and the theorem holds, since $W_\Psi(n,s,t,v,p)\geq d^{(C)}$.
If the min is greatest, we have 
%\begin{align*}
%W_\Psi(n,s,t,0,p)\geq 
\[d^{(A)} = \Omega\lt(\lt(n^{\omega}/p\rt)^{s/\omega}\rt) \quad \text{and} \quad
%W_\Psi(n,s,t,0,p)\geq 
d^{(B)} = \Omega\lt(\lt(n^{\omega}/p\rt)^{t/\omega}\rt),\]
%\end{align*}
furthermore,
\(W_\Psi(n,s,t,0,p) = \Omega\lt(\lt(n^{\omega}/p\rt)^{\max(s,t)/\omega}\rt).\)
Identical arguments can be made for the cases when only $s=0$ and only $t=0$, yielding the formula in the theorem.
%q\lt[\lt(d^{(A)}+\chchoose{n}{s+v}/p\rt)\lt(d^{(B)}+\chchoose{n}{v+t}/p\rt)\lt(d^{(C)}+\chchoose{n}{s+t}/p\rt)\rt]^{1/2}}.\]
\end{proof}

\subsection{Communication Lower Bounds for the Symmetry Preserving Algorithm}

We now derive communication cost lower bounds for the symmetry preserving algorithm.
This algorithm performs less computation and has a lower rank than the direct evaluation algorithm analyzed in the previous section.
%However, we determine that it has a different expansion bound and consequently different communication cost requirements.
We again use Lemma~\ref{thm:tlw} to prove the expansion bound for the symmetry preserving algorithm in order to lower bound the size of the projected sets (inputs and outputs of some bilinear algorithm \rev{restriction}).
However, for the symmetry preserving algorithm, some of these projected sets can have a higher dimensionality than corresponding ones in the direct evaluation algorithm.

We show that both the vertical and horizontal communication cost lower bounds are asymptotically greater for the symmetry preserving algorithm than the direct evaluation algorithm, whenever $0<\min(s,t,v)<\omega/3$ (possible only when $\omega\geq 4$).
These cases correspond to matrix-matrix-like contractions, where \rev{at least one of the two matrices is nonsquare.}
\begin{lemma}\label{lem:prsv_exp}
An expansion bound on $\nsyfsalg stv$ is
\[
{
\mathcal{E}^{(s,t,v)}_\Phi(d^{(A)}{,}d^{(B)}{,}d^{(C)})=\min\lt(\lt({\omega \choose t}d^{(A)}\rt)^{\frac{\omega}{s+v}}{,}\lt({\omega \choose s}d^{(B)}\rt)^{\frac{\omega}{v+t}}{,}\lt({\omega \choose v}d^{(C)}\rt)^{\frac{\omega}{s+t}}\rt).}\]
\end{lemma}
\begin{proof}
For $\nbilalg stvDK= \nsyfsalg stv$, we have the following non-zero structure in the sparse tensors specifying the corresponding bilinear algorithm, $\forall \tpl i\omega \in \enti 1n\omega$,
\begin{alignat*}{2}
& \forall(\tpl j{s+v},\tpl{a}{t})\in \bprt{s+v}t(\tpl{i}{\omega}), 
&& \quad K^{(A)}_{\tpl{j}{s+v}\tpl{i}{\omega}}=\frac{t!}{\rho(\tpl at)},  \\
& \forall(\tpl l{v+t},\tpl{b}{s})\in \bprt{v+t}s(\tpl{i}{\omega}), 
&& \quad K^{(B)}_{\tpl{l}{v+t}\tpl{i}{\omega}}=\frac{s!}{\rho(\tpl bs)},  \\
& \forall(\tpl h{s+t},\tpl{c}{v})\in \bprt{s+t}v(\tpl{i}{\omega}),
&& \quad K^{(C)}_{\tpl{h}{s+t}\tpl{i}{\omega}}=s!t!\rho(\tpl cv).
\end{alignat*}
Consider any $\bilalg R\subseteq \nbilalg stvDK$ and the associated set of products, $Z\subseteq D^{(M)}=\enti{1}{n}\omega$.
We prove that $d^{(A)}\defeq \rank(\bhRA) \geq |Z|^{(s+v)/\omega}/{\omega \choose t}$.
Let $L^{(A)}$ be the number of rows in $\bhRA$ that are nonzero,
\[L^{(A)}= \{\tpl j{s+v} : \tpl j{s+v} \in \prj {s+v}(\tpl i{\omega}), \tpl i{\omega} \in Z \}.\]
Each column of $\bhRA$ has at most ${\omega \choose t}$ nonzeros.
Therefore, since the columns of $\bhRA$ to which set $Z$ is mapped, contain nonzeros in $|L^{(A)}|$ different rows, there exists a subset of these columns of size $|L^{(A)}|/{\omega \choose t}$ where each column has a unique nonzero row, and consequently these columns are linearly-independent.
Such a subset can be found by successively picking columns with a nonzero row index that is not yet in the working set, and adding at most ${\omega \choose t}$ new indices to the working set with each added column.
The existence of this subset of columns implies that $\rank(\bhRA)\geq |L^{(A)}|$.

By Lemma~\ref{thm:tlw}, $|L^{(A)}|\geq |Z|^{(s+v)/\omega}$ and so $\rank(\bhRA) \geq |Z|^{(s+v)/\omega}$.
Identical arguments may be used to show the symmetric bounds $d^{(B)}\defeq \rank(\bhRB) \geq |Z|^{(v+t)/\omega}/{\omega \choose s}$ and $d^{(C)}\defeq \rank(\bhRC) \geq |Z|^{(s+t)/\omega}/{\omega \choose v}$.
Combining the three bounds yields the bound in the lemma.
\end{proof}

\subsubsection{Vertical Communication Lower Bounds for the Symmetry Preserving Algorithm}
\label{sec:seq_comm_prsv}
As done previously, we use the expansion bound to prove a lower  bound on the vertical communication cost on the symmetry preserving algorithm.
A previous proof of these communication lower bound on algorithm $\syfsalg{\C A}{\C B}stv$ appeared in \cite{ES_dissertation_2014}.
It employed the lower bound technique from \cite{christ2013communication}, which relies on the H\"older inequality~\cite{Holder1889} and its generalization~\cite{bennett2008brascamp}.
The proof we provide here is more general as it allows for reuse of partial summations in computing operands of different products.
\begin{theorem}\label{thm:seq_comm_prsv_lowerb}
Let $\kappa \defeq \max(s+v,v+t,s+t)$.
Any schedule of $\nsyfsalg stv$ on $\seqarch{\cs}$ has vertical communication cost,
{\mathsmall
\begin{align*}
\bwcost_\Phi(n,s,t,v,\cachesize) &\geq   \max\Bigg[\frac{2\chchoose{n}{\omega}\cachesize}{\lt(3{\omega \choose \kappa}\cachesize\rt)^{\omega/\kappa}} ,
 \chchoose{n}{s+v}+\chchoose{n}{v+t}+\chchoose{n}{s+t} \Bigg] \\
%                                    \frac{n^{s+t}}{(s+t)!}+ \frac{n^{s+v}}{(s+v)!}+ \frac{n^{v+t}}{(v+t)!}\Bigg), \\
&= \Omega\lt(n^{\omega}/\cs^{\omega/\kappa-1} + n^{\kappa}\rt).
\end{align*}
}
\end{theorem}
\begin{proof}
%By Lemma~\ref{lem:prsv_exp}, we know that
%\[\mathcal{E}^{(s,t,v)}_\Phi(d^{(A)},d^{(B)},d^{(C)})=\min\lt(\lt({\omega \choose t}d^{(A)}\rt)^{\frac{\omega}{s+v}},\lt({\omega \choose s}d^{(B)}\rt)^{\frac{\omega}{v+t}},\lt({\omega \choose v}d^{(C)}\rt)^{\frac{\omega}{s+t}}\rt).\]
Applying Theorem~\ref{thm:seq_comm_lwb_bi}, with the expansion bound $\mathcal{E}^{(s,t,v)}_\Phi$, we obtain the vertical communication lower bound,
\[\bwcost_\Phi(n,s,t,v,\cachesize)\geq \max\lt[\frac{2\chchoose{n}{\omega}\cs}{\mathcal{E}_\Phi^\mathrm{max}(\cs)}, \chchoose{n}{s+v}+\chchoose{n}{v+t}+\chchoose{n}{s+t}\rt],\]
where
\begin{align*}
\mathcal{E}_\Phi^\mathrm{max}(\cs) &\defeq \max_{c^{(A)},c^{(B)},c^{(C)}\in \mathbb{N}, c^{(A)}+c^{(B)}+c^{(C)}= 3\cs}\mathcal{E}_\Phi(c^{(A)},c^{(B)},c^{(C)}).
\end{align*}
Substituting $\mathcal{E}^{(s,t,v)}_\Phi$ from Lemma~\ref{lem:prsv_exp}, we obtain,
\begin{align*}
\mathcal{E}_\Phi^\mathrm{max}(\cs)&\leq \mathcal{E}^{(s,t,v)}_\Phi(3\cs,3\cs,3\cs) =\lt(3{\omega \choose \kappa}\cs\rt)^{\omega/\kappa},
\end{align*}
where ($\kappa \defeq \max(s+v,v+t,s+t)$).
Inserting this upper bound on $\mathcal{E}_\Phi^\mathrm{max}(\cs)$ into the cost lower bound on $\bwcost_\Phi(n,s,t,v,\cachesize)$, yields the cost lower bound shown in the theorem.
%Plugging this in, we obtain the lower bound stated in the theorem.
%\[\bwcost_\Phi(n,s,t,v,\cachesize)\geq \max\lt[\frac{2\chchoose{n}{\omega}\cs}{\mathcal{E}_\Phi^\mathrm{max}(\cs)}, \chchoose{n}{s+v}+\chchoose{n}{v+t}+\chchoose{n}{s+t}\rt],\]
\end{proof}
We note that this vertical commutation lower bound is asymptotically the same as that of the direct evaluation algorithm (Theorem~\ref{thm:comm_lowerb_strd}) when $\kappa=2\omega/3$, or when one of $s,t,v$ is zero.
When $s,t,v>0$ and $\kappa>2\omega/3$ (so $0<\min(s,t,v)<\omega/3$), this communication cost lower bound is asymptotically greater.
% than that of Theorem~\ref{thm:comm_lowerb_strd}.

\subsubsection{Horizontal Communication Lower Bounds for the Symmetry Preserving Algorithm}
\label{sec:par_comm_prsv}
We now derive a horizontal communication lower bound for parallel executions of the symmetry preserving algorithm.
Again the lower bound is obtained directly by substitution of the expansion bound into Lemma~\ref{thm:par_comm_lwb_bi} for general bilinear algorithms.
%Interestingly, the resulting bound relates to the direct evaluation algorithm horizontal communication lower bounds with the same relation as between the pair of vertical communication lower bounds.
%In particular, 
The second horizontal communication lower bound for the direct evaluation algorithm (Theorem~\ref{thm:comm_lowerb_syst_prl2}) is the same as the bound we prove below for the symmetry preserving algorithm when exactly one of $s,t,v$ is zero.
Like the vertical communication lower bound, the horizontal communication lower bound is greater for the symmetry preserving algorithm than for the direct evaluation algorithm, when $0<\min(s,t,v)<\omega/3$.
\begin{theorem}\label{thm:par_comm_prsv_lowerb}
Any storage-balanced schedule of $\syfsalg{\C A}{\C B}stv$ on $\pararch{p}{M}$ has horizontal communication cost,
%$\syfsalg{\C A}{\C B}stv$ %$Phi_\odot^{(s,t,v)}(\C A,\C B)$ 
%on a homogeneous machine with $p$ processors each with local memory $M$
%is 
\begin{enumerate}
\item \(W_\Phi(n,s,t,v,p)= \Omega\lt((n^\omega/p)^{\kappa/\omega}\rt)\) with $\kappa \defeq \max(s+v,v+t,s+t)$, when $s,t,v\geq 1$.
\item \(W_\Phi(n,s,t,v,p)= \Omega\lt((n^\omega/p)^{\max(s,t,v)/\omega}\rt),\)
when exactly one of $s,t,v$ is zero.
%\item \fixed{Edgar}{When two of $s,t,v$ are zero, I don't really care since $\Psi=\Phi$.}
\end{enumerate}
\end{theorem}
\begin{proof}
By Lemma~\ref{lem:prsv_exp}, we have the expansion bound,
\[{\mathcal{E}^{(s,t,v)}_\Phi(d^{(A)},d^{(B)},d^{(C)})=\min\lt[\lt({\omega \choose t}d^{(A)}\rt)^{\frac{\omega}{s+v}},\lt({\omega \choose s}d^{(B)}\rt)^{\frac{\omega}{v+t}},\lt({\omega \choose v}d^{(C)}\rt)^{\frac{\omega}{s+t}}\rt].}\]
Applying Theorem~\ref{thm:par_comm_lwb_bi}, with expansion bound $\mathcal{E}^{(s,t,v)}_\Phi$, we obtain the horizontal communication cost bound,
\[W_\Psi(n,s,t,v,p)\geq d^{(A)}+d^{(B)}+d^{(C)},\]
with the constraint that $d^{(A)},d^{(B)},d^{(C)}\in \mathbb{N}$ and
{\mathfootnotesize
\begin{alignat*}{2}
& \frac{\chchoose{n}{\omega}}{p} \leq 
\mathcal{E}^{(s,t,v)}_\Phi&&\Bigg(d^{(A)}+\chchoose{n}{s+v}/p,d^{(B)}+\chchoose{n}{v+t}/p,d^{(C)}+\chchoose{n}{s+t}/p\Bigg) \\
&=\min\Bigg[\bigg({\omega \choose t}&&\lt(d^{(A)}+\chchoose{n}{s+v}/p\rt)\bigg)^{\frac{\omega}{s+v}},
            \lt({\omega \choose s}\lt(d^{(B)}+\chchoose{n}{v+t}/p\rt)\rt)^{\frac{\omega}{v+t}}, \\
           & &&\lt({\omega \choose v}\lt(d^{(C)}+\chchoose{n}{s+t}/p\rt)\rt)^{\frac{\omega}{s+t}}\Bigg].
\end{alignat*}
}
When $s,t,v\geq 1$, this constraint implies that we must have
\[\min\lt((d^{(A)})^{\frac{\omega}{s+v}},(d^{(B)})^{\frac{\omega}{v+t}},(d^{(C)})^{\frac{\omega}{s+t}}\rt)=\Omega\lt(\chchoose{n}{\omega}/{p}\rt).\]
Furthermore, in this case, for $\kappa = \max(s+v,v+t,s+t)$,
\[W_\Phi(n,s,t,v,p)= \Omega\lt((n^\omega/p)^{\kappa/\omega}\rt).\]
However, when $v=0$, \(\lt({\omega \choose v}\tchchoose{n}{s+t}/p\rt)^{\frac{\omega}{s+t}}=\frac{\chchoose{n}{\omega}}{p}\), so we have only the constraint,
\[\min((d^{(A)})^{\frac{\omega}{s}},(d^{(B)})^{\frac{\omega}{t}})=\Omega\lt(\chchoose{n}{\omega}/{p}\rt).\]
This constraint implies the lower bound,
\[W_\Phi(n,s,t,0,p)= \Omega\lt((n^\omega/p)^{\max(s,t)/\omega}\rt).\]
Similar arguments apply when instead of $v=0$, we have $s=0$ or $t=0$.
Combining the bounds from these three cases yields the general bound stated in the theorem.
%since the terms added to $d^{(A)}$, $d^{(B)}$ and $D^{(C)}$ are of size $\Theta(n^\omega/p^{\frac{\omega}{s+v}})$, $\Theta(i1/p^{\frac{\omega}{v+t}})$, and  $\Theta(1/p^{\frac{\omega}{s+t}})$, rather than $\Theta(1/p)$.
\end{proof}

%\section{Lower Bounds for Nested Bilinear Algorithms}
%\label{sec:nest_lb}
%\input{nest_lb}

%\section{Sequential Lower Bounds}
%\label{sec:seq_lb}
%\input{seq_lb}
%
%\section{Parallel Lower Bounds}
%\label{sec:par_lb}
%\input{par_lb}

\section{Conclusion}
\label{sec:conc}

Table~\ref{tab:lbsum} in Section~\ref{sec:intro} summarizes our lower bounds results in the general case and for various contractions of particular order.
%To draw conclusions from lower bounds, it is important to also consider their attainability.
The matrix multiplication and nonsymmetric tensor contraction vertical communication lower bounds are easily
attainable by \rev{computing products of} $\sqrt{\cs/3} \times \sqrt{\cs/3}$ blocks of each matrix.
The parallel horizontal communication bounds for these nonsymmetric algorithms are also attainable for all dimensions~\cite{demmel2013communication}.
The vertical communication lower bounds for the direct evaluation algorithm, $\systalg ABstv$, are asymptotically attainable by
use of an efficient matrix multiplication algorithm, however, attaining the constant factors
is more difficult. When one of $s,t,v$ is zero, they may be attained by exploiting the symmetry
of the largest tensor for each block loaded in cache. 
%For example, when $t=0$, it makes sense to 
%perform all multiplications for each element of $\C A$ that is read into cache, as 
%reading $\C B$ and writing $\C C$ should in this case have low order costs. When the cache
%is small enough, such an approach should achieve the lower bound with the desired symmetry factor
%of $1/(s+v)!$ rather than $1/(s!v!)$ on the leading order communication cost.
%However, we leave it for future work to determine whether it is possible to attain the lower
%bound for arbitrary $s,t,v$ or if a stronger communication lower bound exists.
The parallel horizontal communication lower bounds for the direct evaluation algorithm should
be asymptotically attainable when $s,t,v>0$ by unpacking the tensors and performing an efficient
matrix multiplication algorithm~\cite{SMHSD_JPDC_2014,Rajbhandari:2014:CFC:2683593.2683635}.
When exactly one of $s,t,v$ is zero, storing the largest tensor in packed layout and moving the 
other operands should asymptotically attain the lower bounds presented in this paper.
%however, we are not aware of an appropriate communication analysis of this case.
%Such an algorithm should be within the space of algorithms considered in existing 
%frameworks~\cite{SMHSD_JPDC_2014,rajbhandari2013framework},
%but we are not aware of a careful cost analysis for such cases. 
%As pointed out previously, 

\rev{Relative to these upper bounds, our lower bound results suggest that in the distributed setting, it may not be worthwhile to store tensors
in packed layout when extra memory is available, as unpacked storage permits lower communication cost for some matrix-vector-like contractions.}
For the symmetry preserving algorithm, $\syfsalg ABstv$, the sequential vertical communication cost can be attained
within a constant factor by the algorithm given in~\cite{ES_dissertation_2014}.
We conjecture that the horizontal communication bounds may also be attained asymptotically for any choice of $(s,t,v)$.
\rev{However, our lower bounds for symmetry preserving algorithms already show that any parallel schedule for this algorithm will require more communication cost than that achievable with the direct evaluation algorithms for some contractions.}
%As future work, it is interesting to investigate the reuse of partial sums as a method for lowering
%the computation and/or the communication cost of the symmetry preserving algorithm.
% or to find an approach to prove
%via lower bounds that partial summations are never useful.

%Our lower bounds demonstrate that the symmetry preserving algorithm has the same asymptotic
%communication bounds as the direct evaluation algorithm for the types of (anti)symmetry in coupled cluster contractions (always matrix-vector-like).
%However, coupled-cluster contractions generally involve partially symmetric tensors, which can be efficiently computed via nested use of the symmetry preserving algorithm~\cite{SD_ETHZ_2015}.
%These algorithms for partially-symmetric contractions are themselves still bilinear algorithms and, moreover, they can be derived via tensor (Kronecker) products as
%\(\bilalg{F} = \bilalg{G}\otimes \bilalg{H},\)
%where $\B{F^{(A)}} = \B{G^{(A)}}\otimes \B{H^{(A)}}$, 
%$\B{F^{(B)}} = \B{G^{(B)}}\otimes \B{H^{(B)}}$, and
%$\C{F^{(C)}} = \C{G^{(C)}}\otimes \C{H^{(C)}}$.
%However, the systematic extension of the bilinear algorithm expansion bounds for nested bilinear algorithms is non-trivial, so we leave it for future work.

\rev{Overall, our analysis makes steps toward characterization of the communication costs necessary for contraction methods of tensors with permutational index symmetry.}
Tight bounds on the communication cost of such contractions would provide communication cost bounds for tensor-contraction-based methods such as coupled-cluster.
Further, the infrastructure we develop for bilinear algorithms is extensible not only to partial-symmetric tensors, but also to contractions of sparse or structured tensors (e.g., Toeplitz or Henkel matrices).
\rev{Our notion of expansion bounds for bilinear algorithms provides a potential route for deriving communication lower bounds for these problems in a unified way.}
%the communication cost of algorithms which nest the symmetry preserving algorithm~\cite{SD_ETHZ_2015} require further analysis.
%Such analysis of nested contraction algorithm would be particularly important for partially symmetric contraction which arise in coupled cluster.

%\section{Acknowledgements}
%\input{ack}
\bibliographystyle{siamplain}
\bibliography{paper}
\closeoutputstream{fixmes}

\section{Appendix I: Volumetric Inequalities}
%\section{Volumetric Inequalities}
\label{sec:vol_ineq}

To derive expansion bounds for algorithms, we will employ volumetric inequalities to lower bound the sizes of sets of projections.
In particular, we use the following generalization of the Loomis-Whitney inequality~\cite{loomis1949n,Tiskin98thedesign}. %, which we state below.
\begin{theorem}%[Generalized Discrete Loomis-Whitney Inequality] 
\label{thm:otlw} 
Let $V$ be a set of $m$-tuples, $V \subseteq \inti{1}{n}^m$,
Consider $m\choose r$ projections:
\[\forall \tpl{s}{r}\in \lnti{1}{m}{r}, \quad \pi_{\tpl{s}{r}}(\tpl{i}{m})\defeq (i_{s_1}, \ldots, i_{s_r}),\]
and apply these projections to all elements in $V$ to form the projection sets:
\[\forall \tpl{s}{r}\in \lnti{1}{m}{r}, \quad  L_\tpl{s}{r} \defeq \{\pi_{\tpl{s}{r}}(\tpl{i}{m}) :  \tpl im \in V\}.\]
The cardinality of the set $V$ may be upper bound by the cardinalities of these projections,
\[  |V| \le \bigg(\prod_{\tpl{s}{r}\in \lnti{1}{m}{r}}|L_\tpl{s}{r}|\bigg)^{1/{m-1\choose r-1}}. \]
\end{theorem}
The standard Loomis-Whitney inequality~\cite{loomis1949n} is given by Theorem~\ref{thm:otlw} with $m=d$ and $r=d-1$.
We use the $d=3$ form of it to prove Lemma~\ref{lem:mm_exp}.
%For matrix multiplication, which computes corresponds an order three tensor of scalar multiplications, Theorem~\ref{thm:otlw} is used with $m=3$ and $r=2$, corresponding to $d=3$.
%Each scalar multiplication may be assigned a tuple $(i,j,k)$ such that it contributes to $C_{ij}$ and requires elements $A_{ik}$ and $B_{kj}$.
%The three such sets of such dependencies of a set $V$ of three-tuples are then given by $L_{(1,2)}$, $L_{(1,3)}$, and $L_{(2,3)}$ in Theorem~\ref{thm:otlw} and yield the familiar bound,
%\[|V|\leq (|L_{(1,2)}||L_{(1,3)}||L_{(2,3)}|)^{1/2}.\]
%This inequality was first used to obtain matrix multiplication communication lower bounds in~\cite{irony_mm_lb04} and we also use it to obtain an improved communication bound for matrix multiplication (see Appendix~\ref{apx:seq_comm_lwb_mm}).
%\fixme{Edgar}{Above is somewhat redundant with respect to what happens in the next section.}

When applying such inequalities to obtain communication lower bounds, we are generally interested in a lower bound on the size of the projected sets $\{L_\tpl{s}{r}\}$, rather than an upper bound on $V$.
So, we introduce the following lemma, which gives a lower bound on the union of the projections $L=\bigcup_{\tpl{s}{r}\in \lnti{1}{m}{r}}L_\tpl{s}{r}$ and can be succinctly expressed using $\prj r$ (Definition~\ref{def:prj}).
\begin{lemma}%[Generalized Discrete Loomis-Whitney Inequality] 
\label{thm:tlw} 
Let $V$ be a set of $m$-tuples, $V \subseteq \inti{1}{n}^m$,
%Consider $m\choose r$ projections 
%\[\pi_{\tpl{s}{r}}(\tpl{i}{m})=(i_{s_1}, \ldots, i_{s_r}), \forall \tpl{s}{r}\in \lnti{1}{m}{r},\]
consider the projected sets given by the projection map $\prj r$,
\[ L = \{ \tpl{w}{r} : \tpl{w}{r}\in \prj r(\tpl{v}{m}), \tpl{v}{m}\in V\},\]
then we have % the cardinality of $V$ is no greater than % of this set of projections is at least
%the size of the set $V$ may be upper bound as
%\[  |V| \le \prod_{\tpl{s}{r}\in \lnti{1}{m}{r} }|\pi_{\tpl{s}{r}}(V)|^{1/{m-1\choose r-1}}. \]
%Conversely, the size of the union of the projections may be lower bound by the size of the set $V$,
$  |V|\le \lt|L\rt|^{m/r}.$
%Let $V$ be a finite, nonempty set of $m$-tuples $\tpl{i}{m} \in \inti{1}{n}^m$.
%Consider $m\choose r$ projections 
%\[\pi_{\tpl{s}{r}}(\tpl{i}{m})=(i_{s_1}, \ldots, i_{s_r}), \forall \tpl{s}{r}\in \lnti{1}{m}{r},\]
%the size of the set $V$ may be upper bound by the sizes of these projections,
%\[  |V| \le \prod_{ \tpl{s}{r}\in \lnti{1}{m}{r} }|\pi_{\tpl{s}{r}}(V)|^{1/{m-1\choose r-1}}. \]
%Conversely, the size of the union of the projections may be lower bound by the size of the set $V$,
%\[  \lt|\{\pi_{\tpl{s}{r}}(V) : \tpl{s}{r}\in \lnti{1}{m}{r} \}\rt| \leq |V|^{r/m}. \]
\end{lemma}

The proof of Lemma~\ref{thm:tlw} may be easily obtained directly from Theorem~\ref{thm:otlw}, see~\cite{ES_dissertation_2014}.

\section{Appendix II: Lower Bounds for Nonsymmetric Contraction Algorithms}
\label{sec:ns_lb}
%\subsection{Vertical Communication Lower Bounds for Matrix Multiplication}
%\label{sec:seq_comm_mm}
%\input{seq_comm_mm}
%
%
%\subsection{Horizontal Communication Lower Bounds for Matrix Multiplication}
%\label{sec:par_comm_mm}
%\input{par_comm_mm}
%
%
%\subsection{Vertical Communication Lower Bounds for Nonsymmetric Contractions}
%\label{sec:seq_ns_de}
%\input{seq_ns_de}
%
%\subsection{Horizontal Communication Lower Bounds for Nonsymmetric Contractions}
%\label{sec:par_ns_de}
%\input{par_ns_de}
We exercise the bilinear algorithm lower bound infrastructure by deriving communication lower bounds for nonsymmetric contractions (matrix multiplication and the general case, which follows trivially).
These results are well-known, although our lower bound constants are stronger than those presented in previous analyses.
All later communication lower bound proofs will follow the same logical structure as the ones in these section.
\subsection{Lower Bounds for Matrix Multiplication}
\label{sec:comm_mm}

We start by applying the theory developed for bilinear algorithms to matrix multiplication, reproducing known results.
As for all bilinear algorithms, we start by deriving the expansion bound of the algorithm, then applying Theorem~\ref{thm:seq_comm_lwb_bi} and Theorem~\ref{thm:par_comm_lwb_bi} to obtain horizontal and communication lower bounds, respectively.
\begin{lemma}\label{lem:mm_exp}
An expansion bound (Definition~\ref{def:biexp}) on the classical (non-Strassen-like) matrix multiplication algorithm of $m$-by-$k$ matrix $\B A$ with $k$-by-$n$ matrix $\B B$ into $m$-by-$n$ matrix $\B C$ is
\[\mathcal{E}_\text{MM}(d^{(A)},d^{(B)},d^{(C)})=(d^{(A)}d^{(B)}d^{(C)})^{1/2}.\]
\end{lemma}
\begin{proof}
For the matrix multiplication bilinear tensor algorithm $\nbilalg 111DK$, where $D^{(A)}=\inti 1m\inti 1k$, $D^{(B)}=\inti 1k\inti 1n$, $D^{(C)}=\inti 1m\inti 1n$, and $D^{(M)}=\inti 1m\inti 1n\inti 1k$ (these compose $\mathcal{D}$), the sparse tensors specifying the corresponding bilinear algorithm are
{\mathfootnotesize
\begin{align*}
\forall i_1\in \inti{1}{m},i_2\in\inti{1}{n},i_3\in\inti{1}{k},
K^{(A)}_{i_1i_3i_1i_2i_3}= 
K^{(B)}_{i_3i_2i_1i_2i_3}= 
K^{(C)}_{i_1i_2i_1i_2i_3}=1. 
\end{align*}
}
Consider any $\bilalg R\subseteq \getmat{\nbilalg 111DK}$ and the associated subset of products, $V\subseteq D^{(M)}$.
%For each $(i_1,i_2,i_3)\in V$, there is a unique nonzero in each of the above three sparse tensors:
%$F^{(A)}_{i_1i_3i_1i_2i_3}$, $F^{(B)}_{i_2i_3i_1i_2i_3}$, $F^{(C)}_{i_1i_2i_1i_2i_3}$.
%This means that 
The columns of the matrices $\B{{R}}^{(A)}$, $\B{{R}}^{(B)}$, and $\B{{R}}^{(C)}$ (contained in $\mathcal{{R}}$) each have a single non-zero (unit) entry and so are only linearly dependent when they are equivalent.
Let the numbers of unique columns in these three matrices be $d^{(A)}=\rank(\B{R}^{(A)})$, $d^{(B)}=\rank(\B{R}^{(B)})$, and $d^{(C)}=\rank(\B{R}^{(C)})$.
The number of such unique columns is also the size of the projection sets $d^{(A)}=|L_{(1,3)}|$,
where $L_{(1,3)}=\{(i_1,i_3) : (i_1,i_2,i_3)\in V\}$ and similarly for $\B B$ and $\B C$.
%These are the same three projections as used in Theorem~\ref{thm:otlw} with $m=2$ and $r=3$,  so 
Therefore, we can apply Theorem~\ref{thm:otlw} with $m=2$ and $r=3$ to bound the cardinality of $V$ as 
{\mathfootnotesize
\[\rank(\bilalg R) = |V|\leq(d^{(A)}d^{(B)}d^{(C)})^{1/2}=(\rank(\B{R}^{(A)})\rank(\B{R}^{(B)})\rank(\B{R}^{(C)}))^{1/2}.\]
}
Thus Definition~\ref{def:biexp} is satisfied for the expansion bound $\mathcal{E}_\text{MM}$.
\end{proof}

%\fixme{Edgar}{Update lead-in sentence.}
We now give a lower bound on the communication cost of matrix multiplication.
This lower bound result is not new from an asymptotic stand-point (the asymptotic lower bound was first proven by \cite{Jia-Wei:1981:ICR:800076.802486}).
%, but serves as a basic demonstration of the proof technique we use for all particular bilinear algorithms.
The first term in the bound is a factor of 16 higher than the lower bound given earlier by \cite{greygeneral2010}, where the assumptions on initial/final data layout and overlap between input entries are looser.
%, although
%the proof argument is not significantly different (both rely on the Loomis-Whitney theorem). % closely resembles the one used in~\cite{greygeneral2010}.
\begin{theorem}\label{thm:seq_comm_lwb_mm}
Any sequential schedule of the classical (non-Strassen-like) matrix multiplication algorithm of $m$-by-$k$ matrix $\B A$ with
$k$-by-$n$ matrix $\B B$ into $m$-by-$n$ matrix $\B C$ on $\seqarch{\cs}$
has vertical communication cost,
$$\bwcost_\mathrm{MM}(m,n,k,\cachesize)\geq \max\lt[\frac{2mnk}{\sqrt{\cachesize}},mk+kn+mn\rt].$$
\end{theorem}
\begin{proof}
By Lemma~\ref{lem:mm_exp}, we know that the classical matrix multiplication algorithm has expansion bound  $\mathcal{E}_\text{MM}(d^{(A)},d^{(B)},d^{(C)})=(d^{(A)}d^{(B)}d^{(C)})^{1/2}$.
Applying Theorem~\ref{thm:seq_comm_lwb_bi}, with this expansion bound, we obtain the communication lower bound,
\[\bwcost_\mathrm{MM}(m,n,k,\cachesize)\geq \max\lt[\frac{2mnk\cs}{\mathcal{E}_\mathrm{MM}^\mathrm{max}(\cs)}, mk+kn+mn\rt],\]
where
\[\mathcal{E}_\mathrm{MM}^\mathrm{max}(\cs) = \max_{c^{(A)},c^{(B)},c^{(C)}\in \mathbb{N}, c^{(A)}+c^{(B)}+c^{(C)}\leq 3\cs}  (c^{(A)}c^{(B)}c^{(C)})^{1/2} = \cs^{3/2},\]
which is strictly increasing and convex for $\cs\geq 1$ as needed, so we arrive at the bound stated in the theorem.
%\[\bwcost_\mathrm{MM}(m,n,k,\cachesize)\geq \max\lt[\frac{2mnk}{\sqrt{\cs}}, mk+kn+mn\rt].\]
\end{proof}
%Since this theorem simply improves the constant on a previous result, we give the proof in Appendix~\ref{apx:seq_comm_lwb_mm}.
%The approach in \fixme{Edgar}{Cite old version or thesis?} took closer accounting of constant factors to attain a slightly tighter bound ($(3/2)^{3/2}\approx 1.84$),
%$$\bwcost_\mathrm{MM}(m,n,k,\cachesize)\geq \max\lt[\frac{2mnk}{\sqrt{\cachesize}},mk+kn+mn\rt].$$
%The above theorem constitutes an improvement
%of the constant factor on the lower bound with respect to the best bound we are aware of.
%In particular, 

%\fixed{Edgar}{Do we really want these constants? Nah.}
The following lower bound, Theorem~\ref{thm:par_comm_lwb_mm} is also proven in \cite{demmel2013communication}.
We give an alternate proof using Theorem~\ref{thm:par_comm_lwb_bi}.
\begin{theorem}\label{thm:par_comm_lwb_mm}
Any storage-balanced schedule of the classical (non-Strassen-like) matrix multiplication algorithm of $m$-by-$k$ matrix $\B A$ with
$k$-by-$n$ matrix $\B B$ into $m$-by-$n$ matrix $\B C$ on $\pararch{p}{M}$ has horizontal communication cost, 
{\mathfootnotesize
\[W_\mathrm{MM}(m,n,k,p)=\Omega\lt(W_\mathrm{O}(\min(m,n,k),\mathrm{median}(m,n,k),\max(m,n,k),p)\rt),\]
}
where
%, with $(l_1,l_2,l_3)\in \Pi((m,n,k))$, $l_1\leq l_2\leq l_3$, has horizontal communication cost:
%Let $d^{(A)}=\max(m,n,k)$, $d^{(B)}=\median(m,n,k)$, and $d^{(C)}=\min(m,n,k)$, 
%Consider two cases
%We break the lower bound on the horizontal communication cost of this matrix multiplication
%into three cases: %$ \geq \bar{W}$ where
%\begin{enumerate}
%\item if $p>\frac{2mnk}{\min(m,n,k)^3}$,
%\[\bar{W}\geq \max\lt[\frac{\sqrt{2}}{4} \frac{mnk}{p\sqrt{M}}-M,
%                                  \lt(\frac{mnk}p\rt)^{2/3}-\frac{mnk}{p\min(m,n,k)}\rt].\]
%\item if $\frac{2mnk}{\min(m,n,k)^3}\geq p > \frac{mnk}{\min(m,n,k)\max(m,n,k)^2}$, 
%\[\bar{W}\geq \lt(\sqrt{\frac 23}-\frac{1}{\sqrt{2}}\rt) \sqrt{\frac{\min(m,n,k)mnk}{p}}.\]
%\item if $\frac{mnk}{\min(m,n,k)\max(m,n,k)^2} \geq p$, 
%\[\bar{W}\geq \frac{7}{50} \frac{mnk}{\max(m,n,k)}.\]
%\end{enumerate}
%\begin{enumerate}
%\item 
%\begin{align*}
\[W_\mathrm{O}(x,y,z,p)= 
\begin{cases}
%\frac{xyz}{p\sqrt{M}} + 
\lt(\frac{xyz}p\rt)^{2/3} & : p>yz/x^2.\\ % $p>\frac{mnk}{\min(m,n,k)^3}$,
x\lt(\frac{yz}{p}\rt)^{1/2} & %\Omega\lt(\sqrt{\frac{\min(m,n,k)mnk}{p}}\rt).
: yz/x^2 \geq p > z/y. \\ %$\frac{mnk}{\min(m,n,k)^3}\geq p > \frac{\min(m,n,k)\max(m,n,k)^2}{mnk}$, 
xy& %\frac{mnk}{\max(m,n,k)}\rt).
: z/y \geq p. %$\frac{\min(m,n,k)\max(m,n,k)^2}{mnk} \geq p$, 
\end{cases}
\]
\end{theorem}
\begin{proof}
By Lemma~\ref{lem:mm_exp}, we know that the classical matrix multiplication algorithm has expansion bound  $\mathcal{E}_\mathrm{MM}(d^{(A)},d^{(B)},d^{(C)})=(d^{(A)}d^{(B)}d^{(C)})^{1/2}$.
Applying Theorem~\ref{thm:par_comm_lwb_bi}, with this expansion bound, we obtain the bound,
\[W_\mathrm{MM}(m,n,k,p)\geq d^{(A)}+d^{(B)}+d^{(C)},\]
for some $d^{(A)},d^{(B)},d^{(C)}\in \mathbb{N}$ such that 
\[\displaystyle{mnk/p \leq [(d^{(A)}+mk/p)(d^{(B)}+kn/p)(d^{(C)}+mn/p)]^{1/2}}.\]
Letting $x=\min(m,n,k)$, $y=\mathrm{median}(m,n,k)$, and $z=\max(m,n,k)$, we rewrite the above as
\[W_\mathrm{O}(x,y,z,p)= d_1+d_2+d_3,\]
for some $d_1,d_2,d_3\in \mathbb{N}$ such that
\[\displaystyle{xyz/p \leq [(d_1+xy/p)(d_2+xz/p)(d_3+yz/p)]^{1/2}}.\]
The symmetry of the objective and constraint in $d_1,d_2,d_3$ tells us that due to $x\leq y \leq z$, the optimal solution must have $d_1\geq d_2 \geq d_3$.
Asymptotically, there are three scenarios that are distinguished by which term in the right-hand side of the constraint is greatest (dominant):
\begin{itemize}
\item $\sqrt{d_1d_2d_3}$ is the dominant term, $W_\mathrm{O}(x,y,z,p)=\Omega((xyz/p)^{2/3})$,
\item $\sqrt{d_1d_2yz/p}$ is the dominant term, $W_\mathrm{O}(x,y,z,p)=\Omega(x\sqrt{yz/p})$, 
\item $(z/p)\sqrt{d_1xy}$ is the dominant term, $W_\mathrm{O}(x,y,z,p)=\Omega(xy)$.
\end{itemize}
Computing the ranges in which each of these three terms implies the least communication yields the lower bounds given in the theorem.
\end{proof}
%\fixme{Edgar}{Ugh... that was harder to prove than I expected and still is a bit hand-wavy, maybe should not bother presenting this since its known? Self-answer: similar things done later, so this is illustrative and probably worth keeping/editing.}
These matrix multiplication lower bounds can be interpreted geometrically.
They correspond to partitioning the cube of $mnk$ products in matrix multiplication in one, two, or three dimensions (1D, 2D, and 3D algorithms).

\subsection{Lower Bounds for Nonsymmetric Contractions}
\label{sec:comm_ns_de}

We now introduce communication lower bounds for the nonsymmetric contraction algorithm. % by extension from matrix multiplication.
As we explained after the statement of Algorithm~\ref{alg:nsctr}, $\nsalg{\C A}{\C B}stv$ is equivalent to a matrix multiplication of an $n^s\times n^v$ matrix with an $n^v\times n^t$ matrix yielding a $n^s\times n^t$ matrix. 
Therefore, its communication lower bounds have a direct correspondence to those of matrix multiplication.
\begin{theorem}\label{thm:comm_lowerb_ns}
Any schedule of $\nsalg{\C A}{\C B}stv$ 
on $\seqarch{\cs}$ has vertical communication cost,
%with a cache of size $\cachesize$,
%i.e., the number of words moved between cache and memory,
%under the assumption that no operands start in cache and all outputs are 
%written to memory and that all elements have unit element size
%has the following bounds,
\begin{align*}
 \bwcost_\Upsilon(n,s,t,v,\cachesize) %\le             & \frac{2n^{\omega}}{\sqrt{\cachesize}}+n^{s+t} + n^{s+v} + n^{v+t} 
            & \geq \max\lt[\frac{2n^{\omega}}{\sqrt{\cachesize}},
                  n^{s+t}+ n^{s+v}+ n^{v+t}\rt].
\end{align*}
\end{theorem}
\begin{proof}
%In~\cite{SD_ETHR_2014} it was shown in the proof of the computation cost of the nonsymmetric
%tensor contraction algorithm, that the computations done by this algorithm have a one to one
Since any matrix multiplication may be expressed as a direct evaluation nonsymmetric contraction algorithm, the communication cost of this nonsymmetric contraction algorithm cannot be lower than the optimal communication cost of the standard matrix multiplication algorithm that evaluates the $n^\omega$ products directly.
Therefore, $\bwcost_\Upsilon(n,s,t,v,\cachesize)\geq \bwcost_\mathrm{MM}(n^s,n^t,n^v,\cachesize)$, which yields the lower bound above.
%The upper bound on the cost may be attained by the classical algorithm that computes matrix multiplication
%block-by-block with blocks of each matrix of dimension $\sqrt{\cachesize/3}$.
%The lower bound is asymptotically the same as the classical result by Hong and 
%Kung~\cite{Jia-Wei:1981:ICR:800076.802486}, but includes constant factors rather than just
%an asymptotic bound. A lower bound with a constant factor was also given by~\cite{greygeneral2010}, 
%but with a smaller constant prefactor. Our proof of the lower 
%bound with this constant factor is in Section~\ref{appndx:comm_lowerb:mm}.
%
%We note that in practice, the ordering of elements in which the matrix is laid out in memory
%may present non-trivial overheads if the desired matrix element ordering is different.
%This challenge is amplified in the case when the tensor contraction is done on a distributed-memory system.
%\renewcommand{\qedsymbol}{}
\end{proof}

Similarly, we obtain a bound for horizontal communication cost below.
%We start with a parallel horizontal communication lower bound for nonsymmetric contractions, which is just that of a matrix multiplication.
\begin{theorem}\label{thm:comm_lowerb_ns_prl}
Any storage-balanced schedule of $\nsalg{\C A}{\C B}stv$ 
on $\pararch{p}{M}$ has horizontal communication cost, %a parallel machine with $p$ processors each with $M$ memory
%i.e., the number of words moved between the processors,
%under the assumptions that the work is load balanced is %\fixme{Edgar}{make this precise} is
%(every processor does $O(F^\Upsilon(\varrho,n,s,t,v)/p)$ work)
%and that all elements have unit element size has the following lower bound,
\begin{align*}
W_\Upsilon(n,s,t,v,p)    =\Omega\lt(W_\mathrm{MM}\lt(n^s,n^t,n^v,p\rt)\rt).
\end{align*}
%\fixed{Edgar}{The above is only true for $m,n,k\gg 1$, do the other cases.}
%\fixed{Edgar}{Prove this with constants using Loomis-Whitney. Nah}
\end{theorem}
\begin{proof}
By the same argument as in the proof of Theorem~\ref{thm:comm_lowerb_ns}, this algorithm can be used
to perform a matrix multiplication with dimensions $n^s$, $n^t$, and $n^v$, hence the bound stated in the theorem.
\end{proof}

%
%
%\subsection{Horizontal Communication Lower Bounds for Matrix Multiplication}
%\label{sec:par_comm_mm}
%\input{par_comm_mm}
%
%
%\subsection{Vertical Communication Lower Bounds for Nonsymmetric Contractions}
%\label{sec:seq_ns_de}
%\input{seq_ns_de}
%
%\subsection{Horizontal Communication Lower Bounds for Nonsymmetric Contractions}
%\label{sec:par_ns_de}
%\input{par_ns_de}

\section{Appendix III: Lower Bounds Direct Evaluation of Symmetric Contractions}
\label{sec:syde_lb}
%We apply a similar lower bound approach for the direct evaluation algorithm of symmetric tensor contractions and obtain communication lower bounds that are somewhat smaller than the nonsymmetric case.
%This result is expected, since the direct evaluation algorithm requires less computation than the nonsymmetric algorithm for most contractions.
%%Further, additional symmetric equivalence between elements exists within these matrices, which makes a lower communication cost possible. 
%However, for matrix-vector-like contractions, we also provide an additional expansion bound for the direct evaluation algorithm that does not have a counterpart for the nonsymmetric algorithm.
In this appendix, we establish matrix-multiplication-like communication lower bounds for the direct evaluation algorithm for symmetric contractions.
We start by deriving an expansion bound for this bilinear algorithm.
Note that we have also derived another expansion bound and horizontal communication lower bound for this algorithm in Section~\ref{subsec:sym_de_lb_main}.
\begin{lemma}\label{lem:strd_exp}
An expansion bound %(Definition~\ref{def:biexp}) 
for $\nsystalg stv$ is
\[\mathcal{E}^{(s,t,v)}_\Psi(d^{(A)},d^{(B)},d^{(C)})=q\lt(d^{(A)}d^{(B)}d^{(C)}\rt)^{1/2},\]
where $q\defeq \lt[{s+v \choose s}{v+t \choose v}{s+t \choose s}\rt]^{1/2}$.
\end{lemma}
\begin{proof}
For the bilinear tensor algorithm $\nbilalg stvDK= \nsystalg stv$, the products are specified by the following sparse tensors: %we have the following non-zero structure in the sparse tensors specifying the corresponding bilinear algorithm,
\begin{alignat*}{5}
\forall &\tpl g{s+v}\in \enti 1n{s+v}, &&\tpl lt\in\enti 1nt, &&(\tpl js, \tpl kv)\in\chi^s_v(\tpl g{s+v}),  
\quad && K^{(A)}_{\tpl{g}{s+v}\tpl{j}{s}\tpl{l}{t}\tpl{k}{v}}=1, \\
\forall &\tpl h{v+t}\in \enti 1n{v+t}, &&\tpl js\in\enti 1ns, &&(\tpl kv, \tpl lt)\in\chi^v_t(\tpl h{v+t}),  
\quad && K^{(B)}_{\tpl{h}{v+t}\tpl{j}{s}\tpl{l}{t}\tpl{k}{v}}=1, \\
\forall &\tpl i{s+t}\in \enti 1n{s+t}, &&\tpl kv\in\enti 1nv, &&(\tpl js, \tpl lt)\in\chi^s_t(\tpl i{s+t}),
\quad && K^{(C)}_{\tpl{i}{s+t}\tpl{j}{s}\tpl{l}{t}\tpl{k}{v}}=s!t!\rho(\tpl kv).
\end{alignat*}
Consider any $\bilalg R\subseteq \getmat{\nbilalg stvDK}$ and the associated subset of products $V\subseteq D^{(M)}$, where $D^{(M)}=\enti{1}{n}s\enti{1}{n}t\enti{1}{n}v$.
%For each $(i_1,i_2,i_3)\in V$, there is a unique nonzero in each of the above three sparse tensors:
%Consider any subset $V\subseteq D^{(M)}=\enti{1}{n}s\enti{1}{n}t\enti{1}{n}v$. 
Like in matrix multiplication, the columns of the matrices $\B{R}^{(A)}$, $\B{R}^{(B)}$, and $\B{R}^{(C)}$ (with elements contained in $\mathcal{{R}}$) 
each have a single non-zero entry (as each product has one operand from $\C A$, one from $\C B$ and contributes to one output in $\C C$) and so are only linearly dependent when they are equivalent.
Consider the three projections:
\begin{align*}
L_{(1,3)} = \{\tpl js\tpl kv : \tpl js\tpl lt \tpl kv\in V\},
L_{(2,3)} = \{\tpl kv\tpl lt : \tpl js\tpl lt \tpl kv\in V\},
L_{(1,2)} = \{\tpl js\tpl lt : \tpl js\tpl lt \tpl kv\in V\},
\end{align*}
The row corresponding to index $\tpl g{s+v}\in \enti 1n{s+v}$ is non-zero in $\B{R}^{(A)}$ if $\exists \tpl js\tpl kv\in L_{(1,3)}$ such that $(\tpl js, \tpl kv)\in \chi^s_v(\tpl g{s+v})$.
Since there is a unique such row $\tpl g{s+v}$ for each entry of $L_{(1,3)}$ and  $|\chi^s_v(\tpl g{s+v})| = {s+v\choose v}$, it follows that 
the number of non-zero rows in $\B{R}^{(A)}$ must be at least $d^{(A)}\defeq \rank(\B{R}^{(A)}) \geq |L_{(1,3)}|/{s+v \choose v}$.
%Subsequently, we have $\rank(\C{{R}^{(A)}})\geq d^{(A)}$ and 
Similarly for $\C B$ and $\C C$ with  $d^{(B)}\defeq \rank(\B{R}^{(B)})\geq |L_{(2,3)}|/{t+v \choose t}$ and $d^{(C)}\defeq \rank(\B{R}^{(C)})\geq |L_{(1,2)}|/{s+t \choose s}$.
By Theorem~\ref{thm:otlw} with $m=2$ and $r=3$, we then obtain 
%$|V|\leq$.
%Thus, we have
\[|V| \leq (|L_{(1,2)}||L_{(1,3)}||L_{(2,3)}|)^{1/2} \leq \bigg({s+v\choose v}d^{(A)}{t+v \choose t} d^{(B)}{s+t \choose s} d^{(C)}\bigg)^{1/2},\]
so Definition~\ref{def:biexp} is satisfied for the expansion bound,
\[\mathcal{E}^{(s,t,v)}_\Psi(d^{(A)},d^{(B)},d^{(C)})= q\lt(d^{(A)}d^{(B)}d^{(C)}\rt)^{1/2}.\]
where $q= \lt[{s+v \choose s}{v+t \choose v}{s+t \choose s}\rt]^{1/2}$.
\end{proof}

\subsection{Vertical Communication Lower Bounds for Direct Evaluation of Symmetric Contractions}
\label{sec:seq_comm_drct}

We have shown that the expansion bound of the direct evaluation algorithm for symmetric contractions is larger than that of nonsymmetric contractions by a constant factor.
We use this expansion bound to derive communication lower bounds that are lower than the nonsymmetric case by a constant factor.
A simple example of when less vertical communication is needed for the symmetric contraction algorithm than the nonsymmetric one, is when all tensors fit into cache. 
In this case, it suffices to read each element of $\C A$ and $\C B$ once and write each element of $\C C$ once, which attains the below lower bound when $s=t=v$ and is more efficient than the respective matrix multiplication.
%\fixme{Edgar}{Previously the theorem gave:
%\[\bwcost_\Psi(n,s,t,v,\cachesize) \geq \frac{1}{q'}\bwcost_\mathrm{MM}\lt(\chchoose ns,\chchoose nt,\chchoose nv,q' \cachesize\rt), \]
%where $q'\defeq \max\lt({s+v\choose s},{v+t\choose t},{s+t \choose s}\rt)$. The following result is a tiny tiny bit stronger.
%}
\begin{theorem}\label{thm:comm_lowerb_strd}
%Let $q\defeq \max\lt({s+v\choose s},{v+t\choose t},{s+t \choose s}\rt)$,
Any sequential schedule of $\systalg{\C A}{\C B}stv$ on $\seqarch{\cs}$ has vertical communication cost,
%%a sequential machine 
%%with a cache of size $\cachesize$,
%%i.e., the number of words moved between cache and memory,
%%under the assumption that no operands start in cache and all outputs are 
%%written to memory and that all elements have unit element size
%%has the following bounds,
%%is
%\[\bwcost_\Psi(n,s,t,v,\cachesize) \geq \frac{1}{q}\bwcost_\mathrm{MM}\lt(\chchoose ns,\chchoose nt,\chchoose nv,q \cachesize\rt), \]
%or if we expand the above and tighten the input/output lower bound slightly, $\bwcost_\Psi(n,s,t,v,\cachesize)$
%%            & \frac{2n^{\omega}}{s!t!v!\sqrt{\cachesize}}&&+\frac{n^{s+t}}{s!t!} + \frac{n^{s+v}}{s!v!} + \frac{n^{v+t}}{t!v!} \geq 
%\[\geq \max\bigg[\frac{1}{q^{3/2}}\frac{2n^{\omega}}{s!t!v!\sqrt{\cachesize}}, 
%               \frac{n^{s+t}}{(s+t)!}+ \frac{n^{s+v}}{(s+v)!}+ \frac{n^{v+t}}{(v+t)!}\bigg].\]

{\mathfootnotesize
\[\bwcost_\Psi(n,s,t,v,\cachesize)\geq \max\lt[\frac{2\chchoose{n}{s}\chchoose{n}{t}\chchoose{n}{v}}{q\sqrt{\cs}}, \chchoose{n}{s+v}+\chchoose{n}{v+t}+\chchoose{n}{s+t}\rt],\]
}
where $q=\lt[{s+v \choose s}{v+t \choose v}{s+t \choose s}\rt]^{1/2}$.
\end{theorem}
\begin{proof}
By lemma~\ref{lem:strd_exp}, we know that
\[\mathcal{E}^{(s,t,v)}_\Psi(d^{(A)},d^{(B)},d^{(C)})=q\lt(d^{(A)}d^{(B)}d^{(C)}\rt)^{1/2}.\]
Applying Theorem~\ref{thm:seq_comm_lwb_bi}, with this expansion bound, we obtain the communication lower bound,
{\mathfootnotesize
\[\bwcost_\Psi(n,s,t,v,\cachesize)\geq \max\lt[\frac{2\chchoose{n}{s}\chchoose{n}{t}\chchoose{n}{v} \cs}{\mathcal{E}_\Psi^\mathrm{max}(\cs)}, \chchoose{n}{s+v}+\chchoose{n}{v+t}+\chchoose{n}{s+t}\rt],\]
}
where
{\mathfootnotesize
\begin{align*}
\mathcal{E}_\Psi^\mathrm{max}(\cs) &= \max_{c^{(A)},c^{(B)},c^{(C)}\in \mathbb{N}, c^{(A)}+c^{(B)}+c^{(C)}\leq 3\cs}  q(c^{(A)}c^{(B)}c^{(C)})^{1/2}  = q\cs^{3/2},
\end{align*}
}
which is strictly increasing and convex for $\cs\geq 1$ as needed.
\end{proof}

\subsection{Horizontal Communication Lower Bounds for Direct Evaluation of Symmetric Contractions}
\label{sec:par_comm_drct}

Using the same expansion bound, we obtain horizontal communication cost lower bounds for the direct evaluation algorithm for symmetric contractions, which again match those of matrix multiplication.
\begin{theorem}\label{thm:comm_lowerb_syst_prl1}
Any storage-balanced schedule of $\systalg{\C A}{\C B}stv$
on $\pararch{p}{M}$ has horizontal communication cost, %a parallel machine with $p$ processors each with $M$ memory
%i.e., the number of words moved between the processors,
%under the assumptions that the work is balanced is %\fixme{Edgar}{make this precise} is
%(every processor does $O(F^\Upsilon(\varrho,n,s,t,v)/p)$ work)
%and that all elements have unit element size has the following lower bound,
\[W_\Psi(n,s,t,v,p)=\Omega\lt(W_\mathrm{O}(n^{\min(s,t,v)},n^{\mathrm{median}(s,t,v)},n^{\max(s,t,v)},p)\rt),\]
where $W_\mathrm{O}$ is the same as given in Theorem~\ref{thm:par_comm_lwb_mm},
\[W_\mathrm{O}(n^x,n^y,n^z,p)= 
\begin{cases}
%\frac{xyz}{p\sqrt{M}} + 
\lt(\frac{n^{x+y+z}}p\rt)^{2/3} & : p>n^{y+z-2x}.\\ % $p>\frac{mnk}{\min(m,n,k)^3}$,
n^x\lt(\frac{n^{y+z}}{p}\rt)^{1/2} & %\Omega\lt(\sqrt{\frac{\min(m,n,k)mnk}{p}}\rt).
: n^{y+z-2x} \geq p > n^{z-y}. \\ %$\frac{mnk}{\min(m,n,k)^3}\geq p > \frac{\min(m,n,k)\max(m,n,k)^2}{mnk}$, 
n^{x+y}& %\frac{mnk}{\max(m,n,k)}\rt).
: n^{z-y} \geq p. %$\frac{\min(m,n,k)\max(m,n,k)^2}{mnk} \geq p$, 
\end{cases}
\]
\end{theorem}
\begin{proof}
By Lemma~\ref{lem:strd_exp}, we have $\mathcal{E}^{(s,t,v)}_\Psi(d^{(A)},d^{(B)},d^{(C)})=q(d^{(A)}d^{(B)}d^{(C)})^{1/2}$.
Applying Theorem~\ref{thm:par_comm_lwb_bi}, with this expansion bound, we obtain the bound,
\[W_\Psi(n,s,t,v,p)\geq d^{(A)}+d^{(B)}+d^{(C)}\]
for some $d^{(A)},d^{(B)},d^{(C)}\in \mathbb{N}$ such that
\[\displaystyle{\frac{\chchoose{n}{s}\chchoose{n}{t}\chchoose{n}{v}}{p} \leq q\lt[\lt(d^{(A)}+\frac{\tchchoose{n}{s+v}}p\rt)\lt(d^{(B)}+\frac{\tchchoose{n}{v+t}}p\rt)\lt(d^{(C)}+\frac{\tchchoose{n}{s+t}}p\rt)\rt]^{1/2}}.\]
This bound is very similar to what we encounter for matrix multiplication in the proof of Theorem~\ref{thm:par_comm_lwb_mm}, so we proceed accordingly.
Let $x=\min(s,t,v)$, $y=\mathrm{median}(s,t,v)$, and $z=\max(s,t,v)$, so that $x\leq y \leq z$ and rewrite the above objective function as
\[W_\mathrm{O}(n^x,n^y,n^z,p)\geq d_1+d_2+d_3,\]
and constraint function as
\[\displaystyle{\frac{\tchchoose{n}{x}\tchchoose{n}{y}\tchchoose{n}{z}}{p} \leq q\lt[\lt(d_1+\frac{\tchchoose{n}{x+y}}p\rt)\lt(d_2+\frac{\tchchoose{n}{x+z}}p\rt)\lt(d_3+\frac{\tchchoose{n}{y+z}}p\rt)\rt]^{1/2}}.\]
When $x,y,z>0$, we again have three possible asymptotically dominant terms in the optimal solution,
\hspace*{-.5in}
\begin{itemize}
\item if $q\sqrt{d_1d_2d_3}$ is the dominant term, \\ $W_\mathrm{O}(n^x,n^y,n^z,p)=\Omega((n^{x+y+z}/p)^{2/3})$,
\item if $q\sqrt{d_1d_2n^{y+z}/p}$ is the dominant term,  $W_\mathrm{O}(n^x,n^y,n^z,p)=\Omega(n^x\sqrt{n^{y+z}/p})$, 
\item if $q(n^z/p)\sqrt{d_1n^{x+y}}$ is the dominant term,  $W_\mathrm{O}(n^x,n^y,n^z,p)=\Omega(n^{x+y})$.
\end{itemize}
\end{proof}

%\section{Appendix IV: Tables}
%\label{sec:tables}
%\input{tables}

\end{document}